\documentclass[reqno]{amsart}
\usepackage{epstopdf}
\usepackage{graphicx}

\usepackage{subfigure,color}
\usepackage{hyperref}
\usepackage{url}
\usepackage{amsthm,amsfonts,amssymb}
\usepackage{enumerate}

\newtheorem{theorem}{Theorem}[section]
\newtheorem{theoremx}{Theorem}
\newtheorem{definition}[theorem]{Definition}
\newtheorem{proposition}[theorem]{Proposition}

\newtheorem{lemma}[theorem]{Lemma}

\newtheorem{remark}[theorem]{Remark}

\begin{document}

\title{On the reduction criterion for random quantum states} 

\author{Maria Anastasia Jivulescu}
\address{M.A.J.: Department of Mathematics,
Politehnica University of Timi\c soara,
Victoriei Square 2, 300006 Timi\c soara, Romania}
\email{maria.jivulescu@upt.ro}

\author{Nicolae Lupa}
\address{N.L.: Department of Mathematics,
Politehnica University of Timi\c soara,
Victoriei Square 2, 300006 Timi\c soara, Romania}
\email{nicolae.lupa@upt.ro}

\author{Ion Nechita}
\address{I.N.: CNRS, Laboratoire de Physique Th\'{e}orique, IRSAMC, Universit\'{e} de Toulouse, UPS, F-31062 Toulouse, France}
\email{nechita@irsamc.ups-tlse.fr}

\date{\today}

\begin{abstract}
In this paper we study the reduction criterion for detecting entanglement of large
dimensional bipartite quantum systems. We first obtain an explicit formula for the moments of a random quantum state to which the reduction criterion has been applied. We show that the empirical eigenvalue distribution of this random matrix converges strongly to a limit that we compute, in three different asymptotic regimes. We then employ tools from free probability theory to study the asymptotic positivity of the reduction operators. Finally, we compare the reduction criterion with other entanglement criteria, via thresholds.
\end{abstract}


\maketitle 

\section{Introduction}

Entanglement is a fundamental concept in quantum information
theory, considering its applications in quantum teleportation,
quantum cryptography, superdense coding or quantum computing \cite{nc}. One of the most challenging open problems in
the field is to characterize  and classify entangled
states. In its full generality, the problem of deciding whether a mixed quantum state is separable or entangled has been proved to be NP-hard  \cite{gur}.

This problem can be mathematically related to positive
maps on $C^*$-algebras since a quantum
state $\rho_{AB}$ is separable if and only if for all positive
maps $\varphi$ between matrix algebras we have that $[\mathrm{id}\otimes
\varphi](\rho_{AB})\geq 0$, where $\mathrm{id}$ is the identity map on some matrix algebra with appropriate dimension (throughout the paper we will identify states with their density matrices).
Although it is hard to give a full
description of all positive maps, this tool allows to define a
class of necessary conditions for separability, known as
\emph{separability criteria} (positive partial transposition,
reduction, generalized partial transposition).
In the literature, probably the most used entanglement detection tool is the Peres-Horodecki \emph{positive partial transposition} (PPT) criterion \cite{hhh96,per}. This criterion corresponds to the choice $\varphi =\mathrm{transp}$,  the transposition map, which is positive, but not completely positive. It is known from Horodecki et al. \cite{hhh96} that the PPT criterion is sufficient only for  $\mathbb{C}^n \otimes \mathbb{C}^k$  quantum systems with $nk\leq 6$.

In this paper, we are interested in the \emph{reduction criterion}, which is given by the following choice of the positive map
$$ \varphi:M_k(\mathbb{C})\rightarrow M_k(\mathbb{C}),\quad
\varphi(X):=\mathrm{I}_k\cdot\mathrm{Tr} X-X,$$
where $\mathrm{I}_{k}\in M_k(\mathbb C)$ is the identity matrix of size $k$ and $\mathrm{Tr}$ is the usual, unnormalized, matrix trace (here  $M_k(\mathbb C)$ denotes the space of all $k\times k$ complex matrices).
The reduction criterion is known to be weaker than the PPT criterion, the two criteria being equivalent
for  $\mathbb{C}^n \otimes \mathbb{C}^2$  quantum systems  \cite{cag} (and, hence,  the reduction
criterion becomes sufficient for entanglement detection for  $\mathbb{C}^2 \otimes \mathbb{C}^2$ and $\mathbb{C}^3 \otimes \mathbb{C}^2$ quantum systems).
Its importance  is given by the  connection
to entanglement distillation: all states that violate
the reduction criterion are \emph{distillable}  \cite{hho}. Recall that a bipartite entangled state is called distillable if a pure maximally entangled state can be obtained, by local operations and classical communication, from many copies of that state. Note that
the question of whether all states violating the Peres-Horodecki
partial transposition criterion can be distilled is still open to
this day. A generalized version of the reduction criterion with
two parameters is given by Albeverio et al. \cite{acf} For more details about separability criteria, we refer the reader to Horodecki et al. \cite{hhh}

The (normalized) Wishart matrices are known to be
physically reasonable models for random density matrices on a tensor product
space.  In this work, we study the asymptotic eigenvalue
distribution of the \emph{reduced  matrix}
$$R=W^{red}:=W_A \otimes \mathrm{I}_k - W_{AB},$$
where $W=W_{AB}$ is a Wishart matrix of
parameters $nk$ and $s$ and $W_{A}$ is its partial trace with respect to the second subsystem $B$. The matrix above is, up to normalization, equal to $[\mathrm{id} \otimes \varphi](\rho_{AB})$, for a suitable ensemble of random density matrices $\rho_{AB}$. The terminology ``reduced matrix'' might be confusing for readers with background in quantum information theory: the reduced matrix $R$ is not to be confused with $W_A$. Note also that we shall always apply the reduction map $\varphi$ on the second factor of the tensor product, whereas in the literature, sometimes the authors consider both reductions simultaneously. Here, we shall call this (stronger) criterion, where one asks that both matrices
$[\mathrm{id} \otimes \varphi](\rho_{AB})$ and $[\varphi \otimes \mathrm{id}](\rho_{AB})$
should be positive, the \emph{simultaneous reduction criterion}.

The program of studying entanglement criteria for random density matrices has been carried out for other cases in previous work: the PPT criterion has been investigated by Aubrun \cite{aub}, Banica and Nechita \cite{bne12}, and Fukuda and \'{S}niady \cite{fs}, the realignment criterion by Aubrun and Nechita \cite{ane}, absolute PPT random states by Collins et al. \cite{cny} and different block-modifications of Wishart matrices by  Banica and Nechita \cite{bne13}. The present work can be seen as a continuation of the above mentioned line of work, initiated by  Aubrun \cite{aub}.

Our main objective is to derive a \emph{threshold} for the reduction criterion, in the following sense: Consider a random mixed quantum state $\rho_{AB} \in M_n(\mathbb C) \otimes M_k(\mathbb C)$ obtained by partial tracing over $\mathbb{C}^s$ a uniformly distributed, pure quantum state $x \in \mathbb C^n \otimes \mathbb C^k \otimes \mathbb C^s$, where the $s$-dimensional space is treated like an inaccessible environment. What is the probability that the random quantum state $\rho_{AB}$ satisfies the reduction criterion $[\mathrm{id}\otimes
\varphi](\rho_{AB})\geq 0$ ? When one (or both) of the system dimensions $n$ and $k$ are large, a \emph{threshold phenomenon} occurs: if $s \sim cnk$ for some constant $c>0$, then there is a \emph{threshold value} $c_{red}$ of the scaling parameter, such that the following holds:
\begin{enumerate}
\item for all $c < c_{red}$, as dimension $nk$ grows, the probability that $\rho_{AB}$ satisfies the reduction criterion vanishes;
\item for all $c > c_{red}$, as dimension $nk$ grows, the probability that $\rho_{AB}$ satisfies the reduction criterion converges to one.
\end{enumerate}

The threshold phenomenon was introduced by Aubrun \cite{aub} to study the PPT criterion.
The main result of our work consists in the exact computation of the threshold for the reduction criterion. Our result can be stated informally as follows (for precise statements, see Theorems \ref{thm:balanced-strong}, \ref{thm:unbalanced-A-strong},  \ref{thm:unbalanced-B} and \ref{thm:positivity-mu-kc}):

\begin{theoremx}
The thresholds for the reduction criterion obtained by applying the reduction map on the second factor of a quantum state $\rho_{AB} \in M_n(\mathbb C) \otimes M_k(\mathbb C)$, where $\rho_{AB}$ is the partial trace over $\mathbb{C}^s$ of a random pure quantum state from $\mathbb C^{nk} \otimes \mathbb C^s$, are as follows:
\begin{enumerate}
\item In the balanced case ($n \to \infty$, $k \sim tn$, $s \sim cnk$) and in the first unbalanced case ($n$ fixed, $k \to \infty$, $s \sim cnk$), the threshold is trivial, $c_{red} = 0$: asymptotically, all random quantum states satisfy the reduction criterion;
\item In the second unbalanced case ($n \to \infty$, $k$ fixed, $s \sim cnk$), the threshold is
$$c_{red} = \frac{(1+\sqrt{k+1})^2}{k(k-1)}.$$
\end{enumerate}

The thresholds for the \emph{simultaneous} reduction criterion are as follows:
\begin{enumerate}
\item In the balanced case ($n \to \infty$, $k \sim tn$, $s \sim cnk$), the threshold is trivial, $c_{red} = 0$: asymptotically, all random quantum states satisfy the simultaneous reduction criterion;
\item In the unbalanced case ($\max(n,k) \to \infty$, $m:=\min(n,k)$ fixed, $s \sim cnk$), the threshold is
$$c_{red} = \frac{(1+\sqrt{m+1})^2}{m(m-1)}.$$
\end{enumerate}
\end{theoremx}

The paper is organized as follows:
In Section \ref{sec:reduction-criterion} we recall
some known facts about the reduction criterion. Section \ref{sec:elements} provides the necessary background
from combinatorics (non-crossing partitions, permutations) and free probability theory. In Section \ref{sec:Wick} we review some elements of the graphical calculus introduced by Collins and Nechita \cite{cne10a,cne11b}, whereas in Section \ref{sec:Wishart}
we discuss the Wishart ensemble and the construction of the induced measures. In Section \ref{sec:moments} we
give a combinatorial formula for the moments of the reduced  matrix. The spectral
behavior of the reduction operator under different asymptotic regimes
is studied in Sections \ref{sec:balanced}, \ref{sec:unbalanced-A} and \ref{sec:unbalanced-B}. Then, in Section \ref{sec:positivity},   we study
the positivity of the support of the
limiting spectral measure.  Finally, in Section
\ref{sec:vs-other} we compare the threshold for the reduction criterion  to  the thresholds for other separability criteria.

\medskip

\noindent \textit{Acknowledgments.}
The authors would like to thank David Reeb for some useful remarks on a preliminary version of the work. The authors also acknowledge the hospitality of the Technische Universit\"at M\"unchen,  Politehnica University of Timi\c soara and the Laboratoire de Physique Th\'{e}orique, Universit\'{e} de Toulouse, where this research was conducted.

The work of M.A.J. and N.L. was supported by a grant of the Romanian National Authority for Scientific Research, CNCS-UEFISCDI, project number PN-II-ID-JRP-2011-2.
I.N.'s research has been supported by the ANR projects {OSQPI} {2011 BS01 008 01} and {RMTQIT}  {ANR-12-IS01-0001-01}.

\section{The reduction criterion}
\label{sec:reduction-criterion}

The reduction criterion was introduced by Cerf et al. \cite{cag} and by Horodecki et al. \cite{hho} and it has been recognized as one of the most important ones in entanglement detection. In particular, its connection to entanglement distillation has been put forward by  Horodecki et al. \cite{hho}

Recall that in quantum mechanics, the state of a system is characterized by \emph{density matrices}, i.e. positive semidefinite matrices of unit trace.
For a bipartite system described by a density matrix living in a tensor product space $\rho=\rho_{AB} \in M_{n}(\mathbb{C})\otimes M_{k}(\mathbb{C}) \cong M_{nk}(\mathbb C)$, define its reduction on the second subspace as
\begin{equation}\label{eq:def-rho-red}
\rho^{red} = \rho_A \otimes \mathrm{I}_k - \rho_{AB},
\end{equation}
where $\rho_{A}=[\mathrm{id}\otimes \mathrm{Tr}](\rho_{AB})$ is the partial trace
over the second subspace.
This map can be also written as
\begin{equation}
\rho^{red} = [\mathrm{id} \otimes \varphi](\rho_{AB}),
\end{equation}
where the reduction map $\varphi : M_k(\mathbb C) \to M_k(\mathbb C)$ is given by
\begin{equation}\label{eq:def-phi}
\varphi(X) = \mathrm{I}_k \cdot \mathrm{Tr} X - X.
\end{equation}

The above mapping produces an \emph{entanglement criterion} for the following reason:
\begin{proposition}
Let $\rho=\rho_{AB} \in M_{n}(\mathbb C) \otimes M_{k}(\mathbb
C)$ be a \emph{separable} quantum state. Then $\rho^{red} \geq 0$.
However, if $\rho_{AB}$ is a rank one entangled state, then
$\rho^{red} \ngeq 0$.
\end{proposition}

Let us point out that in the literature, the reduction criterion sometimes consists on the positivity of \emph{both} reduction matrices: one asks that both $\rho^{red}$  and
\begin{equation}\label{eq:def-tilde-rho-red}
\tilde\rho^{red} = [\varphi \otimes \mathrm{id}](\rho_{AB}) = \mathrm{I}_n \otimes \rho_B - \rho_{AB}
\end{equation}
should be positive, where $\rho_{B}=[\mathrm{Tr}\otimes \mathrm{id}](\rho_{AB})$ is the partial trace
over the first subspace. In this work, we shall only focus on the case where the reduction map $\varphi$ is applied only on the second subsystem and we shall discuss the other case separately.

At the level of images, the reduction criterion is always satisfied:
\begin{lemma}\label{lem:image-red}
Let $X_{AB} \in M_{n}(\mathbb C) \otimes M_{k}(\mathbb C)$
be a positive semidefinite matrix. Then, $$\operatorname{Im}( X_A
\otimes \mathrm{I}_k) \supseteq \operatorname{Im}(X_{AB}).$$
\end{lemma}
\begin{proof}
It is enough to prove that every eigenvector of $X_{AB}$ is in the image of $X_A \otimes \mathrm{I}_k$, so let us assume that $X_{AB} = xx^*$, for some
$$x = \sum_{i=1}^r \sqrt{\lambda_i} e_i \otimes f_i \in \mathbb C^n \otimes \mathbb C^k,$$
with $r\leq \min(n,k)$, $\lambda_i >0$, $\sum_{i=1}^r \lambda_i = 1$ and $\{e_i\}_{i=1}^{n}$, $\{f_j\}_{j=1}^{k}$ being orthogonal families in $\mathbb C^n$ and $\mathbb C^k$, respectively
(here, $x^*$ denotes the adjoint of a vector $x$).
We  have then $$Y:=[\mathrm{id} \otimes \mathrm{Tr}](xx^*) \otimes \mathrm{I}_k = \sum_{i=1}^r \lambda_i e_ie_i^* \otimes \mathrm{I}_k.$$
For $y= \sum_{i=1}^r \lambda_i^{-1/2} e_i \otimes f_i$, we have $Yy = x$ which proves that $x \in \operatorname{Im}(X_A \otimes \mathrm{I}_k)$.
\end{proof}

Recall that for any matrix algebra map $F:M_k(\mathbb C) \to M_k(\mathbb C)$, one defines the \emph{Choi matrix} \cite{cho} of $F$ by
$$C_F = [F \otimes \mathrm{id}](E_k),$$
where $E_k \in M_{k^2}(\mathbb C)$ is the (unnormalized) maximally entangled state
\begin{equation*}
E_k = \sum_{i,j=1}^k e_i e_j^* \otimes e_i e_j^*,
\end{equation*}
with $\{e_i\}_{i=1}^k$ being an orthonormal basis of $\mathbb C^k$. It is known \cite{cho} that the map $F$ is completely positive iff its Choi matrix $C_F$ is positive semidefinite.

Let $C_\varphi$ be the Choi matrix of $\varphi$ defined in \eqref{eq:def-phi}; obviously, $C_\varphi = I_{k^2} - E_k$. This proves that the reduction map $\varphi$ is co-completely positive, i.e. $\psi = \varphi \circ \mathrm{transp}$ is a completely positive map:
\begin{equation}\label{eq:def-psi}
\psi(X) = \mathrm{I}_k \cdot \mathrm{Tr} X - X^t,
\end{equation}
where $X^t=\mathrm{transp}(X)$ is the transpose of the matrix $X$.
Indeed, the Choi matrix of the map $\psi$ is $C_\psi = 2P_{sym} \geq 0$, where $P_{sym}$ is the orthogonal projection on the symmetric subspace of $\mathbb C^k \otimes \mathbb C^k$. The above discussion shows, via Choi matrices, that the map $\psi$ is completely positive, while $\varphi$ is not completely positive.

The reduction criterion is closely related to the famous Peres-Horodecki PPT criterion. For a bipartite  density matrix $\rho=\rho_{AB} \in M_{n}(\mathbb{C})\otimes M_{k}(\mathbb{C})$, one defines its \emph{partial transposition} with respect to the second system $B$,
$$\rho^\Gamma := [\mathrm{id} \otimes \mathrm{transp}](\rho_{AB}).$$
The
reduction criterion is in general weaker than the PPT criterion,
although they are equivalent for  $\mathbb{C}^n \otimes \mathbb{C}^2$ quantum systems \cite{cag}.

Let us now discuss the rank of a reduced matrix.
\begin{proposition}\label{prop:rank}
For a positive semidefinite matrix $X_{AB} \in M_n(\mathbb C) \otimes M_k(\mathbb C)$ of rank $s$, consider $R = [\mathrm{id} \otimes \mathrm{Tr}](X_{AB}) \otimes \mathrm{I}_k - X_{AB}$. Then, the matrix $R$ has rank at most $k^2s$.
\end{proposition}
\begin{proof}
First, put $X_A:=[\mathrm{id} \otimes \mathrm{Tr}](X_{AB})$. One has, via Lemma \ref{lem:image-red}, $\mathrm{Im}( X_A \otimes \mathrm{I}_k) \supseteq \mathrm{Im}(X_{AB})$, so that $\operatorname{rank}(R) \leq \operatorname{rank} (X_A \otimes \mathrm{I}_k) = k \operatorname{rank}(X_A)$. Since  $X_A$ is a sum of $k$ sub-matrices of $X_{AB}$, it follows that $\operatorname{rank}(X_A) \leq k\operatorname{rank}(X_{AB})=ks$, proving the result.
\end{proof}

\section{Some elements of combinatorial free probability theory}
\label{sec:elements}

In this section we recall some basic concepts and
results from the free probability theory and related subjects for
the convenience of the reader and in order to make the paper
self-contained. A good treatment of such topics can be found in Nica and Speicher \cite{nsp}.

\subsection{Non-crossing partitions and permutations}

Let $(M,<)$ be a finite totally ordered set. A \emph{partition}
$\pi$ of $M$ is a family  $\left\{V_1,\ldots,V_m\right\}$ of blocks of $\pi$, i.e.
disjoint nonempty subsets of $M$, whose union is $M$.
The number of all these blocks is denoted by $\#\pi$. A partition
where each block consists of exactly two elements is called a
\emph{pair partition} or a \emph{pairing}.

A \emph{non-crossing partition} of $M$ is a partition $\pi$ with
the property that if  $a<b<c<d$ in $M$ such
that $a,c$ belong to the same block of $\pi$ and $b,d$ belong to
the same block of $\pi$, then $a,b,c,d$ belong all to the same
block of $\pi$.
The set of all non-crossing partitions of $M$ is denoted by
$NC(M)$. In the special case that $M$ is $[p]:=\{1,\ldots,p\}$ for
some positive integer $p$, it is denoted by
$NC(p)$. Since $NC(M)$ depends only on the number of elements and on the order of
$M$ we will use the natural identification $NC(M)\cong NC(|M|)$,
where $|M|$ is the cardinal number of $M$.

A partition $\pi$ of $[p]$, $p\geq 2$, is
non-crossing if and only if at least one block $V$ of $\pi$ is an
interval (there exist $0\leq k\leq p-1$ and $r\geq 1$ with
$k+r\leq p$ such that $V=\{k+1,\ldots,k+r\}$) and $\pi\setminus
V\in NC(\{1,\ldots,p\}\setminus \{k+1,\ldots,k+r\})\cong NC(p-r)$.
Therefore, a non-crossing partition $\pi\in NC(p)$ decomposes
canonically (up to a circular permutation) into
$\pi=\hat{1}_{r}\oplus \pi_{0},$
where $\hat{1}_{r}\in NC(r)$ is the contiguous block of size $r$
and $\pi_0\in NC(\{r+1,\ldots,p\})\cong NC(p-r).$

The symmetric group on a finite  set $M$ will be
denoted by $\mathcal S(M)$. Usually, $M$ is the set $[p]$ for some positive integer
$p$. Addition in $[p]$ is understood modulo
$p$. In this case the symmetric group is denoted by $\mathcal S_p$.
For a permutation $\alpha\in \mathcal S_p$, we denote by
$\#\alpha$ the number of cycles of $\alpha$ and by
$|\alpha|$  the length of $\alpha$, defined as the minimal non-negative integer $k\in\mathbb{N}$
such that $\alpha$ can be written as a product of $k$
transposition.

The notations above are  polymorphic, since $\#(\cdot)$ denotes
both the number of blocks of partitions and  the number of cycles
of permutations,  and $|\cdot|$ is used to denote both the
cardinality of sets and the length of permutations, respectively.
If $b$ is a cycle of a permutation $\alpha$, we simply
write $b\in\alpha$.

For any $\alpha\in \mathcal S_p$, we have
\begin{equation}\label{eq.n1}
|\alpha|=p-\# \alpha.
\end{equation}

Also, it is known that if $\alpha\in \mathcal S_p$ is an arbitrary permutation and
$\tau=(i\,j)\in \mathcal S_p$ is a transposition, $1\leq i,\,j\leq p$, $i\neq
j$, then we have
\begin{equation}\label{eq:tr}
\#(\tau \alpha)=\begin{cases} \#\alpha+1, &\text{if $i$ and $j$
belong to the same cycle of $\alpha$},\\
\#\alpha-1, &\text{if $i$ and $j$ belong to different cycles of
$\alpha$}.
\end{cases}
\end{equation}

In the following, we will denote by $\gamma$ the full cycle
permutation
$\gamma=\gamma_p:=(p\ldots 2\,1)\in \mathcal S_p$.
We also denote by $\mathcal S_{NC(\gamma)}$ the set of all permutations
$\alpha\in \mathcal S_p$ which saturate the triangle inequality,
$$|\mathrm{id}^{-1}\alpha|+|\alpha^{-1}\gamma|=|\mathrm{id}^{-1}\gamma|=p-1,$$
and we say that $\alpha$ lies on a geodesic between the identity permutation $\mathrm{id}$ and the
full cycle $\gamma$. We have
\begin{align*}\label{eq:geo}
\mathcal S_{NC(\gamma)} &=\left\{\alpha\in \mathcal S_p\,:\,
|\alpha|+|\alpha^{-1}\gamma|=p-1\right\}.
\end{align*}
If $\alpha$ belongs to $\mathcal S_{NC(\gamma)}$, we denote it by
$\mathrm{id}\rightarrow\alpha\rightarrow\gamma$ and we say that $\alpha$ is
a \emph{geodesic} permutation.

For a given partition $\pi$ of $[p]$ and $i\in[p]$ we define
$\left(t(\pi)\right)(i)\in [p]$ to be the first element of the
sequence $\gamma(i),\gamma^2(i),\ldots,\gamma^p(i)$ which belongs
to the same block of $\pi$ as $i$. Observe that
$t(\pi)\in \mathcal S_p$ and the number of blocks of $\pi$ corresponds to the number of cycles of $t(\pi)$.
It is well known that
the map $\pi\mapsto t(\pi)$ is an isomorphism of posets between
$NC(p)$ and  $\mathcal S_{NC(\gamma)}$  \cite{b}.
According to this result we shall not distinguish between a
non-crossing partition $\pi \in NC(p)$ and its associated geodesic
permutation $t(\pi) \in \mathcal S_{NC(\gamma)}$.

\subsection{Free probability}

In this section we recall some basic facts about free probability
theory needed for the development of the main results of the
paper.

A \emph{non-commutative probability space} is a pair
$(\mathcal{A},\varphi)$, where $\mathcal{A}$ is an algebra over
$\mathbb{C}$  with unit element $1_{\mathcal{A}}$ and
$\varphi:\mathcal{A}\to\mathbb{C}$ is a linear functional such
that $\varphi(1_{\mathcal{A}})=1.$  An element $a\in\mathcal{A}$
is called a  (\emph{non-commutative}) \emph{random variable}.

A non-commutative probability space $(\mathcal{A},\varphi)$ is a
\emph{$C^{*}$-probability space} if $\mathcal{A}$ is a
$C^*$-algebra with involution $a\mapsto a^*$ and $\varphi$ is positive (i.e. $\varphi(a^* a)\geq
0$ for all $a\in\mathcal{A}$). Such a linear functional $\varphi$
is called a \emph{state}. If $a=a^*$ ($a$ is self-adjoint) in a
$C^*$-probability space, then
the \emph{distribution} of $a$ (or the \emph{law} of $a$), denoted by $\mu_{a}$, is the probability measure on the spectrum of $a$ (which is a compact subset of the real line) given by
\begin{equation*}\label{eq:distr}
\int x^p d\mu_{a}(x)=\varphi(a^p), \text{ for }
p\in\mathbb{N}^{*},
\end{equation*}
where $\mathbb{N}^{*}=\mathbb{N}\setminus \{0\}$ denotes the set of all positive integers.
The number $\varphi(a^p)$, $p\in\mathbb{N}^{*}$, is called the
\emph{$p$-th moment} of $a$. In this case, instead of talking
about the  moments of some non-commutative random
variable $a$, we refer to the probability measure $\mu_{a}$ whose
moments $m_p(\mu_a):=\int x^p d\mu_{a}(x)$ are just the moments of $a$.

In this paper we shall be mostly concerned with the
$C^*$-probability space of random matrices $X\in
L^{\infty-}(\Omega,\mathcal{F},\mathbb{P})\otimes
M_n(\mathbb{C})$, endowed with
the operator norm of
matrices $\parallel \cdot \parallel$.  Here $(\Omega,\mathcal{F}, \mathbb{P})$ is a
classical probability space and
$L^{\infty-}(\Omega,\mathcal{F},\mathbb{P})$ is the space of
complex-valued random variables with all moments finite,
$$ L^{\infty-}(\Omega,\mathcal{F},\mathbb{P}):=\bigcap_{1\leq p<\infty}
L^{p}(\Omega,\mathcal{F},\mathbb{P}) .$$
In this case, the state
$\varphi$  is given by the expectation of the \emph{normalized} trace
\begin{equation*}
\varphi(X)=\mathbb{E}\frac{1}{n} \mathrm{Tr} X.
\end{equation*}

We now come to one of the fundamental concepts in the free probability
theory, namely that of \emph{free independence}. A family $\left\{
\mathcal{A}_{i}\right\}_{i\in I}$ of unital subalgebras in a
non-commutative probability space $(\mathcal{A},\varphi)$ is called \emph{freely
independent} if for any positive integer $n$, indices $i(1)\neq
i(2)$, $i(2)\neq i(3)$, \ldots, $i(n-1)\neq i(n)$ in $I$ and any random variables
$a_{j}\in \mathcal{A}_{i(j)}$  with $\varphi(a_j)=0$ ($j=1,\ldots,
n$), it holds that
$\varphi(a_1\cdots a_n)=0.$
Collections of random variables
are called \emph{freely independent} if the unital subalgebras they generate are freely independent.

For a  probability measure $\mu$ on
the real line  with compact
support, we consider its free cumulants $\kappa_{p}(\mu)$,
$p\in\mathbb{N}^{*}$, given by moment-cumulant formula (see Nica and Speicher
\cite[Lecture 11]{nsp}):
\begin{equation}\label{eq:cum}
m_{p}(\mu)=\sum_{\pi\in NC(p)}\prod_{b\in\pi} \kappa_{|b|}(\mu).
\end{equation}
This is completely analogous to the classical cumulants, the only
difference being that the lattice of all partitions has been
replaced by the lattice of non-crossing
partitions. Clearly, the free cumulants $\kappa_{p}(\mu)$ contain
the same information as the moments of $\mu$.

Given $\mu$ and $\nu$  two probability measures on $\mathbb{R}$
with compact support, it is always possible to find  $a_\mu$ and
$a_\nu$  self-adjoint,  freely independent random variables in some
$C^*$-probability space such that $a_\mu$ has distribution $\mu$
and $a_\nu$ has distribution $\nu$. The \emph{additive free
convolution} of $\mu$ and $\nu$ is denoted by $\mu\boxplus\nu$ and
it is defined as the distribution of $a_{\mu}+a_{\nu}$ (see Nica and Speicher
\cite[Lecture 12]{nsp}). For more regularity properties of the
additive free convolution we refer the reader to the paper of Belinschi \cite{bel}. The
atoms of the additive free convolution of two probability measures
have been described by Bercovici and Voiculescu \cite{bvo}:

\begin{proposition}\label{prop:atoms-free-sum}
Let $\mu$ and $\nu$ be two compactly supported probability
measures on the real line, neither of them a point mass. Then, a
real number $a \in \mathbb R$ is an atom of $\mu \boxplus \nu$ if
and only if $a=b+c$, where $b$ and $c$ are atoms of $\mu$ and
$\nu$, respectively, such that $\mu(\{b\}) + \nu(\{c\}) >1$. Moreover, the mass of the atom
$a$ is
\begin{equation*}
[\mu \boxplus \nu](\{a\}) = \mu(\{b\}) + \nu(\{c\}) -1.
\end{equation*}
\end{proposition}

In this paper, we shall mostly be concerned with \emph{compound free Poisson
distributions} (for more details we refer the reader to the book of Speicher \cite{spe}): given $\lambda> 0$ and  $\nu$  a
compactly supported probability measure on the real line, a
probability measure $\pi_{\nu}$ on $\mathbb{R}$ with the free
cumulants
\begin{equation*}\label{eq:cum.comp}
\kappa_p(\pi_\nu)=\lambda\cdot m_p(\nu), \; p\in\mathbb{N}^{*},
\end{equation*}
is called a compound free Poisson distribution.

As a particular case of compound free Poisson distributions, when the measure $\nu$ is a Dirac mass, we obtain the \emph{free Poisson distributions}. For  $\lambda\geq 0$ and $\alpha\in\mathbb{R}$, the
probability measure given by
\begin{equation*}\label{eq:Poisson}
\nu_{\lambda,\alpha}=\begin{cases}
(1-\lambda)\delta_{0}+{\tilde{\nu}}_{\lambda,\alpha} &
\text{ if } 0\leq \lambda
\leq 1,\\
{\tilde{\nu}}_{\lambda,\alpha} & \text{ if } \lambda>1,
\end{cases}
\end{equation*}
where ${\tilde{\nu}}_{\lambda,\alpha} $ is the measure supported
on the interval $\left[\alpha (1-\sqrt{\lambda})^2,\alpha
(1+\sqrt{\lambda})^2\right]$ with density
\begin{equation*}\label{eq:dens}
d{\tilde{\nu}}_{\lambda,\alpha} (x)=\frac{1}{2\pi\alpha
x}\sqrt{4\lambda\alpha^2-(x-\alpha(1+\lambda))^2} dx.
\end{equation*}
The above measure is called the \emph{free Poisson distribution} with rate $\lambda$ and
jump size $\alpha$.  The free cumulants of this distribution are
given by
\begin{equation*}\label{eq:cum.po}
\kappa_{p}(\nu_{\lambda,\alpha})=\lambda\cdot\alpha^p,\;p\in\mathbb{N}^{*}.
\end{equation*}
In the particular case when $\alpha=1$, the free Poisson distribution
is called the \emph{Mar\v{c}enko-Pastur distribution} of parameter
$\lambda$.

\section{Wick calculus}\label{sec:Wick}

We present next some elements of Wick (or Gaussian) calculus and a
graphical way of computing Gaussian integrals, originally
introduced by Collins and Nechita \cite{cne11b}. Recall that a \emph{Gaussian family} is a set of random variables such that any linear combination of the random variables inside the family has a Gaussian distribution. The main tool for computing averages of Gaussian families is the Wick formula:

\begin{lemma}\label{l.Wick}
If $\left\{X_1, \ldots, X_n\right\}$ is a Gaussian family of
random variables, then
\begin{equation*}
\mathbb{E}\left[X_{i(1)}\cdots X_{i(p)}\right]=\sum_{\pi \text{ pairing of
$[p]$}}\prod_{(r,s)\in \pi} \mathbb{E}\left[X_{i(r)} X_{i(s)}\right],
\end{equation*}
for all $p\in\mathbb{N}^{*}$ and all $1\leq i(1),\ldots, i(p)\leq
n$.
\end{lemma}

The Wick formula is an efficient tool for computing the moments of
Gaussian random matrices, based on a graphical formalism developed
for its application. In the following we briefly present some of
the basic rules of the Gaussian graphical calculus \cite{cne11b}.

A diagram is a collection of decorated boxes and possibly wires connecting the boxes along their decorations and
corresponds to an element in a tensor product space. In graphical
language, a tensor corresponds to a box. Graphically, boxes are represented by rectangles with shaped symbols to characterize
the corresponding vector spaces of the tensor. In
addition, these symbols are white (empty) or black (filled) to denote
the primal or dual space. Tensors contractions are represented graphically by wires connecting these symbols.  The connection is
actually set up between two symbols of the same shape (corresponding thus to vector spaces of the same dimension) and different shadings. There exists a conjugate linear involution on the diagrams, denoted by $*$, which reverts the shading of the decorations. Connections
between white and black symbols corresponds to the canonical map
$\mathbb{C}^n\otimes (\mathbb{C}^n)^*\rightarrow \mathbb{C}$. A diagram such that all decorations are connected by wires corresponds to a complex number. See Figure \ref{fig:graphical-Wick} for some examples of diagrams.

\begin{figure}
\centering
\includegraphics{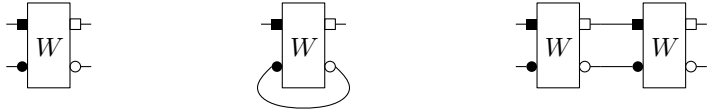}
\caption{From left to right: diagram corresponding to a matrix $W \in M_n(\mathbb C) \otimes M_k(\mathbb C)$, to its partial trace with respect to the second subsystem $W_A = [\mathrm{id} \otimes \mathrm{Tr}](W)$ and to the matrix $W^2$. Square labels correspond to
$\mathbb C^n$ and round labels to $\mathbb C^k$.}
\label{fig:graphical-Wick}
\end{figure}

We now describe how to apply the Wick formula to a diagram  $\mathcal{D}$ containing, among others, boxes $X$ and $\bar{X}$ which correspond to matrices having i.i.d standard, complex Gaussian entries.  A new diagram
$\mathcal{D}_{\alpha}$ is constructed by erasing the boxes $X$ and
$\bar{X}$ and keeping the symbols attached to these boxes. In the new
diagram, the white and black decorations of $i$-th $X$ box are
paired with the decorations of the $\alpha(i)$-th box
$\bar{X}$ in a coherent manner. The resulting diagram
may contain  loops, which correspond to scalars; these scalars are
equal to the dimension of the vector space associated to the
decorations.

In the framework of the graphical calculus, Wick formula can be
reformulated as follows:
\begin{theorem}\label{thm:graphical-Wick}
Let $\mathcal{D}$ be a diagram that contains $p$ boxes $X$ and $p$
boxes $\bar{X}$, which correspond to complex Gaussian $\mathcal{N}(0,1)$
entries. Then
\begin{equation}\label{eq:graphical-Wick}
\mathbb{E}_X[\mathcal{D}]=\sum\limits_{\alpha\in
\mathcal{S}_p}\mathcal{D}_{\alpha}.
\end{equation}
\end{theorem}

\section{Random quantum states and Wishart matrices}
\label{sec:Wishart}

In this short section we recall the classical ensemble of Wishart random matrices and the corresponding induced measures on the set of density matrices.

Let $\phi \phi^*$ be a random pure state on the bipartite Hilbert space $\mathbb{C}^{d}\otimes \mathbb{C}^{s}$, where $\phi$ is a random unit vector distributed \emph{uniformly} on the sphere in $\mathbb{C}^{d}\otimes \mathbb{C}^{s}$.  The \emph{induced measure} of parameters $d$ and $s$ is the distribution $\chi_{d,s}$ of the random density matrix
$$\rho=\mathrm{Tr}_{\mathbb{C}^{s}}\left(\phi\phi^*\right),$$
which is  the partial trace of $\phi\phi^*$ with respect to the environment $\mathbb{C}^{s}$. This measure has been introduced by \.Zyczkowski and  Sommers \cite{zso} and its asymptotical properties have been studied by Nechita \cite{nec}.

Let now $X \in M_{d \times s}(\mathbb{C})$ be a random complex Ginibre matrix (a
$d\times s$ matrix with i.i.d. complex Gaussian $\mathcal{N}(0,1)$
entries). The positive semidefinite matrix $W=XX^* \in M_d(\mathbb C)$ is
called  a \emph{Wishart matrix} of parameters $d$ and $s$. Density matrices can be obtain from Wishart matrices \cite{nec}. More precisely, if $W$
is a Wishart matrix of parameters $d$ and $s$, then
\begin{equation*}
\rho = \frac{W}{\mathrm{Tr} W},
\end{equation*}
has distribution $\chi_{d,s}$.

The main goal of this work is to study the spectral properties of $\rho^{red}$, where $\rho=\rho_{AB}$ is a random quantum state having distribution $\chi_{d,s}$, when $d=nk$ is the dimension of the tensor product space $\mathbb{C}^n\otimes \mathbb{C}^k$.
Since the reduction criterion is invariant by rescaling by positive constants, we shall investigate from now reduced Wishart matrices instead of reduced random density matrices:
\begin{equation*}
R:=W^{red} = W_A \otimes \mathrm{I}_{k} - W_{AB}.
\end{equation*}
The main advantage of this approach is that the distribution of $W=W_{AB}$ is much easier to deal with than the probability distribution $\chi_{d,s}$.

In the following sections, we shall focus on the distribution of the eigenvalues of the random matrix $R$, and, ultimately, on evaluating the probability that the matrix $R$ is positive semidefinite.

\section{Moment formula}
\label{sec:moments}

Let $X \in M_{n k \times s}(\mathbb{C})$ be a random  Ginibre matrix. In this section we shall prove an explicit combinatorial formula for the moments of the random matrix $R = W^{red}\in M_{nk}(\mathbb C)$.
We use the graphical Wick formula introduced in Section \ref{sec:Wick}.
In this language, the matrix $R = W_A \otimes \mathrm{I}_k - W_{AB}$ (where $W_{AB} = XX^*$) is presented in Figure \ref{fig:R}.
\begin{figure}
\centering
\includegraphics{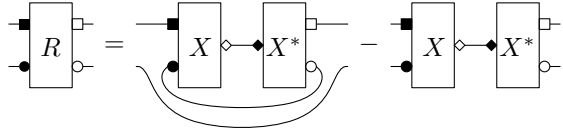}
\caption{Graphical representation of the reduced matrix $R =
(XX^*)^{red} \in M_{nk}(\mathbb C)$. Square labels correspond to
$\mathbb C^n$, round labels to $\mathbb C^k$ and diamond labels
correspond to $\mathbb C^s$.} \label{fig:R}
\end{figure}

Before giving the main result, let us introduce the following notation, which will be essential in what follows:
\begin{definition}\label{def:Pf} \rm
Let $\mathcal F_p$ be the set of all functions $f:\{1,\ldots, p\} \to \{1,2\}$. For a function $f \in \mathcal F_p$, let $P_f \in \mathcal S_p$ be the permutation which behaves like the identity on the set $f^{-1}(1)$ and like $\gamma$ on the set$f^{-1}(2)$.  More precisely,
\begin{equation*}
P_f(i) = \begin{cases}
i, \quad & \text {if } f(i)=1 \text{ or } f^{-1}(2)=\{i\},\\
i-r, \quad & \text {if } f(i)=f(i-r)=2 \text{ and } f(i-1) = \cdots = f(i-r+1) = 1,
\end{cases}
\end{equation*}
where the arguments of $f$ and the values of $P_f$ should be
understood modulo $p$ (identifying $j$ to $p+j$ if
$j\in\{-(p-1),\ldots,0\}$).
\end{definition}

In a similar manner, the definition of $P_f$ can be extended to
functions $f:M \to \{1,2\}$ defined on a totally ordered set
$M=\{a_1,\ldots,a_p\}$  with $a_1<a_2<\cdots < a_p$, identifying
$a_i$ to $i$, $i\in [p]$. In this case we have that
$P_f\in \mathcal S(M)$.

Let us start with some simple examples: for $p=1$, we have
$P_{f(1)=1} = P_{f(1)=2} = (1).$
For $p=2$, we get
\begin{center}
  \begin{tabular}{ | c | c | c | }
    \hline
    $f(1)$ & $f(2)$ & $P_f \in \mathcal S_2$ \\ \hline \hline
    $1$ & $1$ & $(1)(2)$ \\ \hline
    $1$ & $2$ & $(1)(2)$ \\ \hline
    $2$ & $1$ & $(1)(2)$\\ \hline
    $2$ & $2$ & $(2\,1)$\\ \hline
  \end{tabular}
\end{center}
Taking for example  $p=7$ and $f:\{1,\ldots, 7\}\to \{1,2\}$ with
$$f(1)=f(2)=f(4)=f(7)=1 \text{ and } f(3)=f(5)=f(6)=2,$$ we have
$P_f=(1)(2)(4)(7)(6\,5\,3).$

In the following lemma we give the basic properties of the
permutation $P_f$:

\begin{lemma}\label{lem.Pf}
For each $f\in\mathcal{F}_{p}$, we have
\begin{enumerate}
\item[(i)] $P_{f}$ is a geodesic permutation,
\item[(ii)] $\#P_{f}=|f^{-1}(1)|+1-{\bf{1}}_{f\equiv 1}$,
\item[(iii)] $\#\left(P_{f}^{-1}\gamma\right)=p-|f^{-1}(1)|+{\bf{1}}_{f\equiv
1}$,
\end{enumerate}
where ${\bf{1}}_{f\equiv 1}$ denotes a quantity which is equal to
$1$ when $f(i)=1$ for all $i\in [p]$, and $0$ otherwise.
\end{lemma}
\begin{proof}
(i) If there exist $i_1,\ldots, i_k\in [p]$ such that $f(i_1)=\cdots
=f(i_k)=1$, then the partition associated to the permutation $P_f$
is given by
\begin{equation}\label{eq:part.Pf1}
\pi_{P_f}=\left\{ V_1=\{i_1\},\ldots,
V_k=\{i_k\},V_{k+1}=\{1,\ldots,p\}\setminus \{i_1,\ldots,
i_k\}\right\}.
\end{equation}
In the particular case when $f(i)=1$ for all $i\in [p]$, we have
\begin{equation}\label{eq:part.Pf2}
\pi_{P_f}=\left\{ V_1=\{1\},\ldots, V_p=\{p\}\right\}.
\end{equation}
Otherwise, if $f(i)=2$ for all $i\in [p]$, then
\begin{equation}\label{eq:part.Pf3}
\pi_{P_f}=\left\{V=\{1,\ldots,p\}\right\}.
\end{equation}
It follows that $\pi_{P_f}$ is a non-crossing partition and using
the isomorphism between the geodesic permutations and non-crossing partitions,
we have that $P_f$ is a geodesic permutation.
Combining (\ref{eq:part.Pf1})--(\ref{eq:part.Pf3}), we obtain relation (ii).

(iii) Since $P_f$ is a geodesic permutation, we have that
$ |P_f|+|P_f^{-1}\gamma|=p-1.$
By (\ref{eq.n1}), this is equivalent to
$ p-\#P_f+p-\#(P_f^{-1}\gamma)=p-1.$
Using (ii) and relation above, it follows that (iii) holds.
\end{proof}

We can state now the main result of this section:

\begin{theorem}\label{thm:moments}
The moments of the random matrix $R \in M_{nk}(\mathbb C)$ are given by
\begin{equation}\label{eq:moments-R}
\forall p \geq 1, \quad \mathbb E \mathrm{Tr} (R^p) = \sum_{\alpha \in \mathcal S_p, \, f \in \mathcal F_p} (-1)^{|f^{-1}(2)|} s^{\#\alpha}n^{\#(\gamma^{-1}\alpha)}k^{{\bf 1}_{f \equiv 1} + \#(P_f^{-1}\alpha)},
\end{equation}
where the function $f:\{1, \ldots, p\} \to \{1,2\}$ encodes the choice of the term in each factor in the product (choose the $f(i)$-th term in the $i$-th factor)
\begin{equation*}
R^p = (W_A \otimes \mathrm{I}_k - W_{AB})(W_A \otimes \mathrm{I}_k - W_{AB}) \cdots (W_A \otimes \mathrm{I}_k - W_{AB}).
\end{equation*}
Note that in the case when $f \equiv 1$ (only partial traces over $\mathbb{C}^k$) one needs to add an extra factor of $k$, which corresponds to the indicator function.
\end{theorem}
\begin{proof}
The proof is a straightforward application of the graphical Wick formula from Theorem \ref{thm:graphical-Wick}. The first step is to develop $R^p$ using the non-commutative binomial formula:
$$R^p = \sum_{f \in \mathcal F_p} (-1)^{|f^{-1}(2)|} R_f,$$
where $R_f$ denotes the ordered product
$$R_f = R_{f(1)} R_{f(2)} \cdots R_{f(p)} = \overrightarrow{\prod_{1 \leq i \leq p}} R_{f(i)},$$
for the two possible values of the factors
$R_1 =W_A \otimes \mathrm{I}_k$ and $R_2 =W_{AB}.$

We shall now use the formula \eqref{eq:graphical-Wick} to compute $\mathbb E \operatorname{Tr}R_f$. Let us first treat the case when $f \equiv 1$, i.e. when all the factors are equal to $R_1$ above. Before taking the expectation, the diagram for $\operatorname{Tr} R_{f \equiv 1}$ is depicted in Figure \ref{fig:Rf-all-1}.
\begin{figure}
\centering
\includegraphics{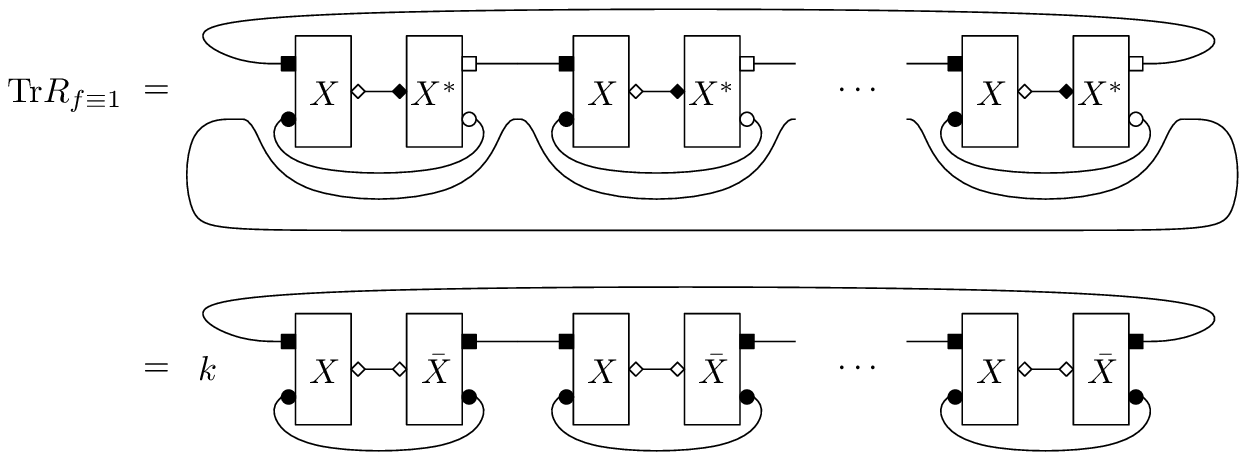}
\caption{Diagram for $\operatorname{Tr} R_{f \equiv 1}$.}
\label{fig:Rf-all-1}
\end{figure}
Using the graphical Wick expansion formula, we write
$$\mathbb E \operatorname{Tr} R_{f \equiv 1} = \sum_{\alpha \in \mathcal S_p} \mathcal D_\alpha,$$
where $\mathcal D_\alpha$ is the diagram obtained by erasing the $X$ and $\bar X$ boxed from Figure \ref{fig:Rf-all-1} and connecting the decorations of the $i$-th $X$ box with the corresponding decorations of the $\alpha(i)$-th $\bar X$ box, as in Figure \ref{fig:Rf-all-1-block-i}.
\begin{figure}
\centering
\includegraphics{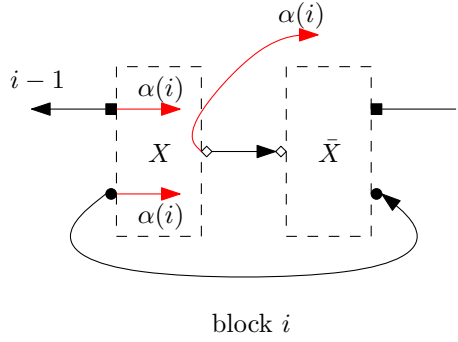}
\caption{Wick diagram expansion: the term corresponding to a permutation $\alpha$ is obtained by erasing the boxes and adding additional wiring according to $\alpha$.}
\label{fig:Rf-all-1-block-i}
\end{figure}

In this way, the resulting diagram $\mathcal D_\alpha$ is a collection of loops (see Figure \ref{fig:Rf-all-1-block-i}):
\begin{enumerate}
\item $\#\alpha$ loops of dimension $s$, corresponding to diamond-shaped decorations. The initial (black) wiring is given by the permutation $\mathrm{id}$ and the additional (red) wiring is given by $\alpha$;
\item $\#\alpha$ loops of dimension $k$, corresponding to round-shaped decorations. The initial (black) wiring is given by the permutation $\mathrm{id}$ and the additional (red) wiring is given by $\alpha$;
\item $\#(\gamma^{-1}\alpha)$ loops of dimension $n$, corresponding to square-shaped decorations. The initial (black) wiring is given by the permutation $\gamma$ ($i \mapsto i-1$) and the additional (red) wiring is given by $\alpha$.
\end{enumerate}
Putting everything together, we get (here, $f \equiv 1$)
\begin{align*}
\mathbb E \operatorname{Tr} R_{f \equiv 1} &= k\sum_{\alpha \in \mathcal S_p} s^{\#\alpha}n^{\#(\gamma^{-1}\alpha)}k^{\#\alpha}\\
&= \sum_{\alpha \in \mathcal S_p} s^{\#\alpha}n^{\#(\gamma^{-1}\alpha)}k^{{\bf 1}_{f \equiv 1} + \#(P_f^{-1}\alpha)}.
\end{align*}

A general term $f \in \mathcal F_p$, not identically equal to $1$, is treated in a similar manner. First, note that there is no longer a factor of $k$ coming from a ``detached'' loop corresponding to square decorations, since there is at least one index $i$ for which $f(i)=2$. For blocks $i$ such that $f(i)=1$, the discussion is the same as in the previous case, the wiring being identical. Blocks with $f(i)=2$ deserve special attention, see Figure \ref{fig:Rf-2}.

\begin{figure}
\centering
\includegraphics{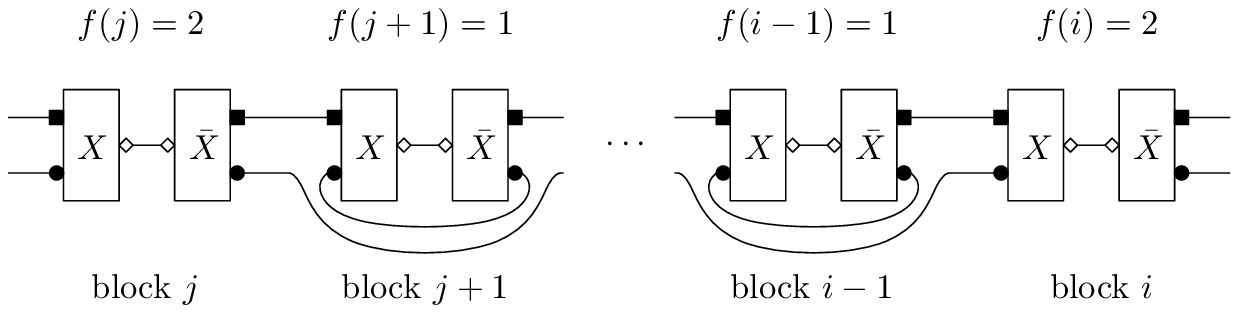}
\caption{Wick diagram expansion for a block with $f(i)=2$. The square black decoration is connected initially to the $\bar X$ box from the block $j=P_f(i)$. All the blocks strictly between $j$ and $i$ are such that $f(j+1) = \cdots = f(i-1)=1$.}
\label{fig:Rf-2}
\end{figure}

The only difference from the situation depicted in Figure \ref{fig:Rf-all-1-block-i} is that the black square label is connected initially to the black label of the $j=P_f(i)$-th $\bar X$ box, with $P_f$ from Definition \ref{def:Pf} (note that it could happen that $j=i$, in the case when $f^{-1}(2) = \{i\}$). The loop counting is identical for diamond-shaped and round-shaped decorations, and we get, for $f$ not identically equal to $1$,
$$\mathbb E \operatorname{Tr} R_{f} = \sum_{\alpha \in \mathcal S_p} s^{\#\alpha}n^{\#(\gamma^{-1}\alpha)}k^{ \#(P_f^{-1}\alpha)}.$$

Putting the two cases together and summing over all $f \in \mathcal F_p$, we obtain the announced formula \eqref{eq:moments-R}.
\end{proof}

As a direct application of the above general formula, the first two moments of $R$ are given by
\begin{align}
\mathbb E \mathrm{Tr} R &= nk(k-1)s, \label{eq:tr1}\\
\mathbb E \mathrm{Tr} \left( R^2 \right) &= (k-2)\left[ (ks)^2n +
ks n^2 \right] + nks^2 + (nk)^2s. \label{eq:tr2}
\end{align}

Understanding the behavior of the combinatorial powers of $n$, $k$ and $s$ in equation \eqref{eq:moments-R} will prove to be key in what follows.

\section{Balanced asymptotics}
\label{sec:balanced}

In this section, we analyze the spectral behavior of the random matrix $R$ in the
``balanced'' asymptotic regime, when both $n$ and $k$ grow, with
linear relative speed.

\noindent\textbf{Balanced asymptotics:} there exist positive constants $c,t>0$ such that
\begin{align}
\label{eq:regime-balanced-first}n &\to \infty; \\
k &\to \infty, \quad k/n \to t;\\
\label{eq:regime-balanced-last}s &\to \infty, \quad s/(nk) \to c.
\end{align}

In this asymptotic regime, we show that the spectrum of the reduced Wishart matrix $R$ becomes trivial when $n \to \infty$, in the sense that $R/(ks) \approx \mathrm{I}_{nk}$. Hence, one can not obtain violations of the reduction criterion by analyzing the global properties of the spectrum of $R$.

\begin{proposition}\label{prop:balanced-moments}
In the balanced asymptotical regime \eqref{eq:regime-balanced-first}--\eqref{eq:regime-balanced-last}, the moments of the rescaled random matrix $R$ converge to $1$:
\begin{equation*}
\forall p \geq 1, \quad \lim_{n \to \infty} \mathbb E \frac{1}{nk}
\mathrm{Tr} \left(\frac{R}{ks}\right)^p = 1.
\end{equation*}
In other words, the empirical eigenvalue distribution
\begin{equation*}
\mu_n = \frac{1}{nk} \sum_{i=1}^{nk}
\lambda_i\left(\frac{R}{ks}\right)
\end{equation*}
converges, in moments, to the Dirac mass at $1$, $\delta_1$.
\end{proposition}
\begin{proof}
Plugging in the asymptotics from equations \eqref{eq:regime-balanced-first}--\eqref{eq:regime-balanced-last} into the moment formula \eqref{eq:moments-R}, we obtain
\begin{equation*}
\mathbb E \mathrm{Tr} (R^p) = (1+o(1))\sum_{\alpha \in \mathcal S_p, \, f \in \mathcal F_p} (-1)^{|f^{-1}(2)|} (ct)^{\#\alpha} t^{{\bf 1}_{f \equiv 1} + \#(P_f^{-1}\alpha)} n^{2\#\alpha + \#(\gamma^{-1}\alpha) + {\bf 1}_{f \equiv 1} + \#(P_f^{-1}\alpha)},
\end{equation*}
for $p\geq 1$.
Let us study the exponent of $n$ in the above relation and try to maximize it in order to find the dominant term (regardless of the sign):
\begin{align*}
\text{exponent of } n &=
2\#\alpha + \#(\gamma^{-1}\alpha) + {\bf 1}_{f \equiv 1} + \#(P_f^{-1}\alpha)\\
&=  {\bf 1}_{f \equiv 1} + 4p-(|\alpha| + |\gamma^{-1}\alpha|) - (|\alpha| + |P_f^{-1}\alpha|)\\
&\leq {\bf 1}_{f \equiv 1} + 4p - |\gamma| - |P_f|\\
&\leq 1+3p+1 = 3p+2,
\end{align*}
where we have used the inequalities
\begin{align*}
|\alpha| + |\gamma^{-1}\alpha| &\geq |\gamma| = p-1,\\
|\alpha| + |P_f^{-1}\alpha| &\geq |P_f| \geq 0,\\
{\bf 1}_{f \equiv 1} &\leq 1,
\end{align*}
which are simultaneously saturated if and only if $\alpha = \mathrm{id}$ and $f \equiv 1$. Thus, there is only one dominating term, and we get
\begin{equation*}
\forall p \geq 1, \quad \mathbb E \mathrm{Tr} (R^p) = (1+o(1)) n^{3p+2}(ct)^pt^{1+p}
\end{equation*}
and the conclusion follows by properly renormalizing $R$ and the trace.
\end{proof}

Let us now make some remarks about Proposition \ref{prop:balanced-moments} and its relation to violations of the reduction criterion. First, note that we only prove convergence in moments of the (random) empirical eigenvalue distribution. This does not imply that there do exist any negative eigenvalues of $R$. What Proposition \ref{prop:balanced-moments} shows is that there do not exist, on average, a finite, strictly positive, fraction of eigenvalues away from $1$. What we show, is that all states pass the reduction entanglement criterion, but in a weak sense: the empirical eigenvalue distribution converges, in moments, to a positively supported probability measure, $\delta_1$. In Theorem \ref{thm:balanced-strong}, we show that a stronger convergence holds, which will allow us to settle the question of the asymptotic positivity of the random matrix $R$. The result makes use of norm-convergence results for Wishart matrices, that we recalled in Appendix \ref{A-convergence}.

\begin{theorem}\label{thm:balanced-strong}
For every $\varepsilon > 0$, the following norm convergence holds:
$$\lim_{n \to \infty} \mathbb P\left(\left\|\frac{R}{ks} - \mathrm{I}_{nk}\right\| > \varepsilon \right) = 0.$$
In particular, the threshold for the reduction criterion in the balanced asymptotical regime is trivial, $c_{red}=0$: with large probability, almost all quantum states satisfy the reduction criterion.
\end{theorem}
\begin{proof}
We start by estimating the quantity in the statement
\begin{align*}
\left\|\frac{R}{ks} - \mathrm{I}_{nk}\right\| &= \left\|\frac{W_A \otimes \mathrm{I_k}}{ks} -\frac{W_{AB}}{ks} - \mathrm{I}_{nk}\right\|\\
&\leq \frac{1}{k}\left\|\frac{W_{AB}}{s}\right\| + \left\|\left(\frac{W_A}{ks} - \mathrm{I}_{n}\right) \otimes \mathrm{I_k}\right\|\\
&\leq \frac{1}{k}\left\|\frac{W_{AB}}{s}\right\| + \left\|\frac{W_A}{ks} - \mathrm{I}_{n}\right\|,
\end{align*}
where $W_A = [\operatorname{id} \otimes \operatorname{Tr}](W_{AB}) \in M_n(\mathbb C)$ denotes the partial trace over the second subspace. We shall use the two results cited in Appendix \ref{A-convergence} to deal with the two terms in the last inequality. First, using Proposition \ref{prop:Bai-Yin}, we obtain
$$ \text{almost surely,} \qquad \lim_{n \to \infty}\left\|\frac{W_{AB}}{s}\right\| = (\sqrt c +1)^2$$
so the first term vanishes. To deal with the second term, recall that $W_A$ is the sum of the diagonal blocks $W_{ii}$ of $W_{AB}$, so
$$\left\|\frac{W_A}{ks} - \mathrm{I}_{n}\right\| = \left\|\sum_{i=1}^k\frac{W_{ii}}{ks} - \mathrm{I}_{n}\right\| \leq \sum_{i=1}^k\frac{1}{k}\left\|\frac{W_{ii}}{s} - \mathrm{I}_{n}\right\|.$$
Note that the random matrices $W_{ii}\in M_n(\mathbb C)$ are {i.i.d.} Wishart matrices of parameter $s \sim cnk \gg n$, so, by Proposition \ref{prop:Z-d} with $\varepsilon=1$, for each $1 \leq i \leq k$, we have that
$$\mathbb P\left(\sqrt\frac{s}{n}\left\|\frac{W_{ii}}{s} - \mathrm{I}_{n}\right\|>3\right) \leq C\exp(-an^{1/3}),$$
for some fixed positive constants $C,a$. In particular, since $s/n \to \infty$, for $n$ large enough, we have that
\begin{equation*}\label{eq:exp-bound-P}
\mathbb P\left(\left\|\frac{W_{ii}}{s} - \mathrm{I}_{n}\right\|>\varepsilon\right) \leq C\exp(-an^{1/3}).
\end{equation*}
We bound now
\begin{align*}
\mathbb P\left(\sum_{i=1}^k\frac{1}{k}\left\|\frac{W_{ii}}{s} - \mathrm{I}_{n}\right\|>\varepsilon\right) &\leq \mathbb P\left(\max_{1 \leq i \leq k}\left\|\frac{W_{ii}}{s} - \mathrm{I}_{n}\right\|>\varepsilon\right)\\
&= 1-\prod_{i=1}^k\mathbb P\left(\left\|\frac{W_{ii}}{s} - \mathrm{I}_{n}\right\| \leq \varepsilon\right)\\
&= 1-\left[1-\mathbb P\left(\left\|\frac{W_{11}}{s} - \mathrm{I}_{n}\right\| > \varepsilon\right)\right]^k\\
&\leq 1-\left[1-C\exp(-an^{1/3})\right]^{k}.
\end{align*}

We can conclude, using  the fact that $C\exp(-an^{1/3})k \to 0$ as $n\to\infty$ and $k\sim tn$.
\end{proof}

From the above result, one can easily infer the threshold for the simultaneous reduction criterion:

\begin{proposition}\label{prop:simultaneous-balanced}
The threshold for the \emph{simultaneous} reduction criterion in the balanced asymptotical regime is trivial, $c_{red}=0$: with large probability, almost all quantum states satisfy the simultaneous reduction criterion.
\end{proposition}
\begin{proof}
The result follows from applying Theorem \ref{thm:balanced-strong} twice. We first apply it the usual setting to obtain that
$$\lim_{n \to \infty} \mathbb P\left(\left\|\frac{R}{ks} - \mathrm{I}_{nk}\right\| > \varepsilon \right) = 0.$$
Then, writing $\tilde R = \mathrm{I}_n \otimes W_B - W_{AB}$ (see also Eq. \eqref{eq:def-tilde-rho-red}) and applying Theorem \ref{thm:balanced-strong} with $t:=1/t$, we get
$$\lim_{n \to \infty} \mathbb P\left(\left\|\frac{\tilde R}{ks} - \mathrm{I}_{nk}\right\| > \varepsilon \right) = 0.$$
We conclude that
$$\lim_{n \to \infty} \mathbb P\left(R \geq 0 \text{ and } \tilde R \geq 0 \right) = 1.$$
\end{proof}

\section{Unbalanced asymptotics, first case}
\label{sec:unbalanced-A}

In this section and in the next one, we focus on \emph{unbalanced} asymptotical regimes, where the smallest of $n$ and $k$ is being kept fixed, while the largest dimension grows to infinity. We start with the case where $n$ is fixed, while $k$, the dimension of the space on which the reduction map is applied, grows.

\noindent\textbf{Unbalanced asymptotics, first case:} there exists a positive constant $c>0$ such that
\begin{align}
\label{eq:regime-unbalanced-A-first}n &\text{ is fixed}; \\
k &\to \infty;\\
\label{eq:regime-unbalanced-A-last}s &\to \infty, \quad s/(nk) \to c.
\end{align}

\begin{proposition}\label{prop:unbalanced-A-moments}
In the first unbalanced asymptotical regime \eqref{eq:regime-unbalanced-A-first}--\eqref{eq:regime-unbalanced-A-last}, the  moments of the rescaled random matrix $R$ converge to $1$:
\begin{equation*}
\forall p \geq 1, \quad \lim_{k \to \infty} \mathbb E \frac{1}{nk} \mathrm{Tr} \left(\frac{R}{ks}\right)^p = 1.
\end{equation*}
In other words, the empirical eigenvalue distribution
\begin{equation*}
\mu_k = \frac{1}{nk} \sum_{i=1}^{nk} \lambda_i\left(\frac{R}{ks}\right)
\end{equation*}
converges, in moments, to the Dirac mass at $1$, $\delta_1$.
\end{proposition}
\begin{proof}
The proof follows closely the one of Proposition \ref{prop:balanced-moments}. After replacing $s$ by $cnk(1+o(1))$ into formula \eqref{eq:moments-R}, we obtain
\begin{equation*}
\forall p \geq 1, \quad \mathbb E \mathrm{Tr} (R^p) = (1+o(1))\sum_{\alpha \in \mathcal S_p, \, f \in \mathcal F_p} (-1)^{|f^{-1}(2)|} c^{\#\alpha}n^{\#\alpha+\#(\gamma^{-1}\alpha)}k^{{\bf 1}_{f \equiv 1} + \#\alpha + \#(P_f^{-1}\alpha)}.
\end{equation*}
The power of the growing parameter $k$ above can be bounded as follows:
\begin{align*}
\text{exponent of } k &=
{\bf 1}_{f \equiv 1} + \#\alpha + \#(P_f^{-1}\alpha)\\
&=  {\bf 1}_{f \equiv 1} + 2p- (|\alpha| + |P_f^{-1}\alpha|)\\
&\leq {\bf 1}_{f \equiv 1} + 2p - |P_f|\\
&\leq 1+2p,
\end{align*}
with equality at every step if and only if $\alpha = P_f = \mathrm{id}$ and $f \equiv 1$. From this we obtain the dominating term and thus
\begin{equation*}
\forall p \geq 1, \quad \mathbb E \mathrm{Tr} (R^p) = (1+o(1))k^{2p+1}c^pn^{p+1},
\end{equation*}
which, after the proper renormalization  allows to conclude (note that $ks \sim cnk^2$).
\end{proof}

One can make the same remarks as for the balanced case: the above result shows that all states satisfy the reduction criterion, in the current asymptotic regime, in a  weak sense (on average, for limiting empirical eigenvalue distributions). It is not excluded that negative eigenvalues remain undetected by our method of moments. As before, we need a stronger type of convergence to conclude.

\begin{theorem}\label{thm:unbalanced-A-strong}
For every $\varepsilon > 0$, the following norm convergence holds:
\begin{equation*}
\lim_{k \to \infty} \mathbb P\left(\left\|\frac{R}{ks} - \mathrm{I}_{nk}\right\|
> \varepsilon \right) = 0.
\end{equation*}
In particular, the threshold for the reduction criterion in the first unbalanced asymptotical regime is trivial, $c_{red}=0$: with large probability, almost all quantum states satisfy the reduction criterion.
\end{theorem}
\begin{proof}
Since $\| M \| ^{2}\leq \mathrm{Tr}\left(M^{2}\right)$ for
every Hermitian matrix $M$, we have, by
Markov's inequality,
\begin{align*}
\mathbb P\left(\left\|
\frac{R}{ks}-\mathrm{I}_{nk}\right\|>\varepsilon\right)&=\mathbb P\left(\left\|
\frac{R}{ks}-\mathrm{I}_{nk}\right\|^2>\varepsilon^2\right) \\
&\leq \mathbb P\left(\mathrm{Tr} \left(
\frac{R}{ks}-\mathrm{I}_{nk}\right)^2>\varepsilon^2\right)\\
&\leq\frac{\mathbb{E}\mathrm{Tr}\left(
\frac{R}{ks}-\mathrm{I}_{nk}\right)^2}{\varepsilon^2}.
\end{align*}
Using relations \eqref{eq:tr1} and \eqref{eq:tr2}, we have
\begin{align*}
\mathbb{E}\mathrm{Tr}\left(
\frac{R}{ks}-\mathrm{I}_{nk}\right)^2&=\frac{1}{k^2s^2}\mathbb{E}\mathrm{Tr}(R^2)-\frac{2}{ks}\mathbb{E}\mathrm{Tr}R+nk\\
&=\frac{2n^2}{s}-\frac{2n^2}{ks}+\frac{n}{k},
\end{align*}
which goes to zero as $k\to\infty$ and $s \sim cnk$. The claim now follows from the
relations above.
\end{proof}

\section{Unbalanced asymptotics, second case}
\label{sec:unbalanced-B}

In this section, we study the second unbalanced asymptotical regime, where $k$ is fixed and $n \to \infty$. Recall that this corresponds to the situation where one of the two subsystems of the partition (called $B$, on which the reduction map is applied) has fixed dimension $k$ and the other one $A$ has a very large dimension $n \to \infty$. As it turns out, we will obtain a non-trivial asymptotic eigenvalue distribution, whose support will be analyzed in the next section.

\noindent\textbf{Unbalanced asymptotics, second case:} there exists a positive constant $c>0$ such that
\begin{align}
\label{eq:regime-unbalanced-B-first}n &\to \infty;\\
k &\text{ is fixed}; \\
\label{eq:regime-unbalanced-B-last}s &\to \infty, \quad s/(nk) \to c.
\end{align}

\begin{theorem}\label{thm:unbalanced-B}
In the second unbalanced asymptotical regime \eqref{eq:regime-unbalanced-B-first}--\eqref{eq:regime-unbalanced-B-last}, the moments of the rescaled random matrix $R$ converge to the following combinatorial quantity:
\begin{equation}\label{eq:asympt-moments-unbalanced-B}
\forall p \geq 1, \quad \lim_{n \to \infty} \mathbb E \frac{1}{nk} \mathrm{Tr} \left(\frac{R}{n}\right)^p = \sum_{\alpha \in NC(p)} \prod_{b \in \alpha} c\left[ (1-k)^{|b|} + k^2-1 \right].
\end{equation}
In other words, the empirical eigenvalue distribution
\begin{equation*}
\mu_n = \frac{1}{nk} \sum_{i=1}^{nk} \lambda_i\left(\frac{R}{n}\right)
\end{equation*}
converges, in moments, to a \emph{compound free Poisson distribution} $\mu_{k,c} = \pi_{\nu_{k,c}}$, where
\begin{equation*}
\nu_{k,c} = c\delta_{1-k} + c(k^2-1) \delta_1.
\end{equation*}
Moreover, the above convergence holds in a \emph{strong} sense: the extremal eigenvalues of the random matrix $R/n$ converge, almost surely when $n \to \infty$, to the edges of the support of the limiting measure $\mu_{k,c}$.
\end{theorem}
\begin{proof}
Let us start by computing the asymptotic moment formula \eqref{eq:asympt-moments-unbalanced-B}. To do this, plug equations \eqref{eq:regime-unbalanced-B-first}--\eqref{eq:regime-unbalanced-B-last} into the general moment formula \eqref{eq:moments-R} to get
\begin{equation*}
\forall p \geq 1, \quad \mathbb E \mathrm{Tr} (R^p) = (1+o(1))\sum_{\alpha \in \mathcal S_p, \, f \in \mathcal F_p} (-1)^{|f^{-1}(2)|} c^{\#\alpha}n^{\#\alpha + \#(\gamma^{-1}\alpha)}k^{{\bf 1}_{f \equiv 1} +\#\alpha+ \#(P_f^{-1}\alpha)}.
\end{equation*}
The exponent of $n$ in the equation above is
\begin{equation*}
\#\alpha + \#(\gamma^{-1}\alpha) = 2p-(|\alpha| + |\gamma^{-1}\alpha|) \leq 2p - |\gamma| = p+1,
\end{equation*}
with equality if and only if $\alpha$ lies on the geodesic between the identity permutation $\mathrm{id}$ and the full cycle $\gamma$: $\mathrm{id} \to \alpha \to \gamma$. Hence, dropping the vanishing terms, we get
\begin{equation*}
\mathbb E \mathrm{Tr} (R^p) = (1+o(1))n^{p+1}\sum_{\alpha \in NC(p), \, f \in \mathcal F_p} (-1)^{|f^{-1}(2)|} c^{\#\alpha}k^{{\bf 1}_{f \equiv 1} +\#\alpha+ \#(P_f^{-1}\alpha)}.
\end{equation*}

Using the formula for $\#(P_f^{-1}\alpha)$ proved in Lemma \ref{lem:Pfm-alpha}, we have
\begin{align}
\mathbb E \mathrm{Tr} (R^p) &= (1+o(1))n^{p+1}k\sum_{\alpha \in NC(p), \, f \in \mathcal F_p} (-1)^{|f^{-1}(2)|} c^{\#\alpha}k^{p-|f^{-1}(1)|+2\sum_{b \in \alpha} \mathbf{1}_{f_b  \equiv 1}} \notag\\
&= (1+o(1))n^{p+1}k\sum_{\alpha \in NC(p)} \sum_{f \in \mathcal F_p} (-1)^{\sum_{b \in \alpha}|f_b^{-1}(2)|} c^{\sum_{b \in \alpha} 1}
k^{\sum_{b \in \alpha}\left(|b|-|f_b^{-1}(1)|+2 \mathbf{1}_{f_b  \equiv 1}\right)} \notag\\
\label{eq:final-E-Tr-Rp} &=
(1+o(1))n^{p+1}k\sum_{\alpha \in NC(p)} \prod_{b
\in \alpha} \sum_{f_b \in \mathcal F_{|b|}}(-1)^{|f_b^{-1}(2)|}
ck^{|b|-|f_b^{-1}(1)|+2 \mathbf{1}_{f_b  \equiv 1}},
\end{align}
where we denote by $f_b$ the restriction of a function $f \in \mathcal F_p$ to a block $b$ of $\alpha$. We compute now
\begin{align*}
S_b &= \sum_{f_b \in \mathcal F_{|b|}}(-1)^{|f_b^{-1}(2)|} ck^{|b|-|f_b^{-1}(1)|+2 \mathbf{1}_{f_b  \equiv 1}}\\
&= ck^{|b|}\left[ k^{2-|b|} - k^{-|b|} + \sum_{f_b \in \mathcal F_{|b|}}(-1)^{|f_b^{-1}(2)|} k^{-|f_b^{-1}(1)|}  \right] \\
&= ck^{|b|}\left[ k^{2-|b|} - k^{-|b|} + \left( \frac{1}{k}-1\right)^{|b|} \right] \\
&= c\left[ (1-k)^{|b|} + k^2-1 \right].
\end{align*}
Plugging the last expression into \eqref{eq:final-E-Tr-Rp}, we obtain
\begin{equation*}
\mathbb E \mathrm{Tr} (R^p) = (1+o(1))n^{p+1}k
\sum_{\alpha \in NC(p)} \prod_{b \in \alpha} c\left[ (1-k)^{|b|} +
k^2-1 \right],
\end{equation*}
which is equivalent to Eq. \eqref{eq:asympt-moments-unbalanced-B} announced in the statement.
We conclude that the empirical eigenvalue distribution of the
random matrix $R/n$ converges, in moments, to a measure
$\mu_{k,c}$, having moments
\begin{equation*}
\int x^p d\mu_{k,c}(x) = \sum_{\alpha \in NC(p)} \prod_{b \in \alpha} c\left[ (1-k)^{|b|} + k^2-1 \right].
\end{equation*}
One identifies the above expression with the moment-cumulant formula
\eqref{eq:cum}, hence the free cumulants of the probability measure $\mu_{k,c}$ are
\begin{equation*}
\kappa_p(\mu_{k,c}) = c\left[ (1-k)^{p} + k^2-1 \right].
\end{equation*}
In the right hand side of the above expression, one recognizes the moments of the measure $\nu_{k,c} = c \delta_{1-k} + c(k^2-1)\delta_1$, and thus $\mu_{k,c}$ is a compound free Poisson distribution,  $\mu_{k,c} = \pi_{\nu_{k,c}}$.

\medskip

We show now the \emph{strong convergence} result. The idea here is to use the general theory developed by Male \cite{mal}, which builds up on the seminal strong convergence result of Haagerup and Thorbj{\o}rnsen for polynomials in independent GUE matrices \cite{hth}.

For a fixed basis $\{e_i\}_{i=1}^k$ of $\mathbb C^k$, denote by $E_{ij}$ the block-matrix units
$$E_{ij} = \mathrm{I}_n \otimes e_ie_j^* \in M_n(\mathbb C) \otimes M_k(\mathbb C).$$

For the sake of simplicity, we denote $W_{AB}$ by $W$.
The random matrix $W/n$ and the constant matrices $E_{ij}$ satisfy the hypotheses of Corollary 2.2 from the paper of Male \cite{mal}. Moreover, one can write the reduced matrix $R$ as a polynomial in $(W, \{E_{ij}\})$, as follows:
$$R = \sum_{i \neq j} W_{jj}^{ii} - \sum_{i \neq j} W_{ij}^{ij},$$
where $W_{ij}^{xy}$ is the block matrix having the $(i,j)$-block $W_{ij}\in M_{n}(\mathbb{C})$ of $W$ in position $(x,y)$ and zeros everywhere else (the indices here run from $1$ to $k$),
$$W_{ij}^{xy} = E_{xi} W E_{jy}=W_{i j}\otimes e_{x} e_{y}^{*}.$$
Indeed, we have
\begin{align*}
R&=\sum_{i} W_{i i} \otimes I_{k}-\sum_{i,j} W_{i j}\otimes e_{i} e_{j}^{*}
= \sum_{i} W_{i i} \otimes \left( \sum_{j} e_{j} e_{j}^{*}\right)-\sum_{i,j} W_{i j}\otimes e_{i} e_{j}^{*}\\
&= \sum_{i\neq j} W_{i i}\otimes e_{j} e_{j}^{*}-\sum_{i \neq j} W_{i j}\otimes e_{i} e_{j}^{*}
= \sum_{i \neq j} W_{jj}^{ii} - \sum_{i \neq j} W_{ij}^{ij}.
\end{align*}
Hence, we obtained the desired strong convergence: the extremal eigenvalues of $R/n$ converge, as $n \to \infty$, to the extreme points of the support of the measure $\mu_{k,c}$.
\end{proof}

\begin{remark}
Similar convergence results have been obtained by Banica and Nechita \cite{bne12, bne13} in the case of the partial transposition map (with respect to the second system) and, respectively, in the more general case of block-modified Wishart matrices. In the setting of the latter work, the (non-unital) reduction map $\varphi$ defined in \eqref{eq:def-phi} acting on the blocks of a Wishart matrix $W$ is not of the type of those obtained by Banica and Nechita  \cite{bne13} in Theorem 4.3, so the combinatorial derivation above was necessary.
\end{remark}

\section{Support positivity and the reduction threshold}
\label{sec:positivity}

In this  section we shall study properties of the limiting measure $\mu_{k,c}$ appearing in the previous section. The main result here is a criterion for the positivity of the support of $\mu_{k,c}$. Note that in the definition of $\mu_{k,c}$ as a compound free  Poisson measure, one can consider $k$ to be an arbitrary real number and we shall assume $k \geq 2$.
\begin{theorem}\label{thm:positivity-mu-kc}
Let $k,c \in \mathbb R$ satisfying $k \geq 2$ and $c > 0$. The probability measure $\mu_{k,c}$ has the following properties:
\begin{enumerate}
\item It has at most one atom, at $0$, of mass $\max(0, 1-ck^2)$.

\item Its support is contained in $(0,\infty)$ if and only if \begin{equation}\label{eq:threshold-red}
c>c_{red} := \frac{(1+\sqrt{k+1})^2}{k(k-1)}.
\end{equation}
\end{enumerate}
\end{theorem}
\begin{proof}

(1) To study atoms of $\mu_{k,c}$, write $\mu_{k,c} = \pi'
\boxplus \pi''$, where $\pi'$ is a free Poisson distribution with rate $c$ and jump size  $1-k$
and
$\pi''$ is a Mar\v{c}enko-Pastur distribution of parameter $c(k^2-1)$.
The probability distributions $\pi'$ and $\pi''$ have at most one atom,
at $0$, of respective masses $1-c$, $1-c(k^2-1)$. From Proposition
\ref{prop:atoms-free-sum} it follows that $\mu_{k,c}$ can have at
most one atom, at $0$, of mass $\max(0,1-c+1-c(k^2-1)-1)$ and the
conclusion follows.

(2) Let us now study the properties of the support of the
absolutely continuous part of $\mu_{k,c}$. The main tool here will
be the Cauchy transform $G$ of $\mu_{k,c}$, for which we will
derive an implicit equation. Start from the $R$-transform of
$\mu_{k,c}$ and derive the $K$-transform \cite{vdn}
\begin{align*}
\mathcal{R}(z) &= \sum_{p=0}^\infty \kappa_{p+1}(\mu_{k,c})z^p = \frac{c(k^2-1)}{1-z} - \frac{c(k-1)}{1+(k-1)z}\\
K(z) &= \mathcal{R}(z) + \frac{1}{z} = \frac{1}{z} + \frac{c(k^2-1)}{1-z} - \frac{c(k-1)}{1+(k-1)z}.
\end{align*}
As the inverse of $K$, $K(G(z))=G(K(z))=z$, the Cauchy transform $G$ of $\mu_{k,c}$ satisfies a degree $3$ polynomial equation
\begin{equation*}
\frac{1}{G(z)} + \frac{c(k^2-1)}{1-G(z)} - \frac{c(k-1)}{1+(k-1)G(z)} = z.
\end{equation*}

We follow now closely the method used in the paper of Banica and Nechita \cite{bne12}. The support
of the absolutely continuous part of $\mu_{k,c}$ is a union of
disjoint intervals; the endpoints of these intervals are the
points on the real line where the analyticity of the Cauchy
transform $G$ breaks. These points are also roots of the
discriminant $\Delta$ of the polynomial equation satisfied by $G$.

With the help of the \textsc{Mathematica} computer algebra system \cite{mat}, we find
\begin{equation*}\label{eq:delta-G}
\Delta =k( \alpha z^4 + \beta z^3 + \gamma z^2 + \delta z + \varepsilon),
\end{equation*}
where
\begin{align*}
\alpha &= k \\
\beta &= 2[k(k-2) - 2c(k-1)^2(k+1)]\\
\gamma &=    2 c^2 k( k-1)^2  ( 3 k^2-4) -
   c (6 k^4 - 8 k^3 - 4 k^2 + 18k -12) + k (k^2-6k+6) \\
\delta &= -2 ( k -1)  [2 c^3 k^2 (k+1)(k-1)^3  -c^2 k (3k^4+k^3-8k^2-6k+10)\\
  \nonumber &\qquad      +c (k^4-k^3-k^2+6k-6)+k(k-2)]\\
\varepsilon &=( k- 1)^2  (c k^2-1)^2 [c^2 k  (k-1)^2  -
   2 c (k^2+k-2)+k].
\end{align*}
We would first like to understand the number of real solutions of the  degree 4 equation $\Delta(z)=0$, i.e. twice the number of intervals of the support of $\mu_{k,c}$. The nature of the roots of a quartic is given by the sign of its discriminant: the discriminant is negative iff the equation has exactly 2 real and 2 complex solutions. For the equation above, the discriminant reads
\begin{equation*}
\Delta_2 = -256 c^2  k^2 (k+1)(k-1)^3 \cdot f^3,
\end{equation*}
where
\begin{equation*}
f = 8 k(k-1)^3 c^3  +  3  (k-1)^2 (5 k^2-9) c^2+ 6  k^3(k-1) c- k^4,
\end{equation*}
hence the sign of the discriminant $\Delta_2$ of the equation $\Delta(z)=0$ is the opposite of the sign of $f$. Let us consider, once more, the discriminant of the cubic equation $f(c)=0$:
\begin{equation*}
\Delta_f = -78732  k^4  (k+1)^2 (k-1)^8,
\end{equation*}
which is negative and thus the equation $f(c)=0$ admits a unique real solution $c=c_0(k)$. Moreover, since $f(0) = -k^4 <0$ and $f(c)\to \infty$ as $t\to\infty$, this solution must be positive, $c_0(k)>0$, for all $k>1$. When $k\to 1^+$ or $k \to \infty$, the equation degenerates, and
\begin{equation*}
\lim_{k \to 1^+} c_0(k) = \infty, \qquad \lim_{k \to \infty} c_0(k) = \frac 1 8.
\end{equation*}
One can compute explicitly the solution of the cubic
\begin{equation*}
c_0(k) = \frac{3 k^3-k^2(5 v-3 u+9)+ 3k(2u-1)+  9(v-u+1)}{8 k (k-1)v},
\end{equation*}
where
\begin{align*}
u &= \sqrt[3]{(k-1) (k+1)^2}\\
v &= \sqrt[3]{(k-1)^2 (k+1)}.
\end{align*}

To sum up, we have shown that the absolutely continuous part of the support of $\mu_{k,c}$ contains two intervals if $c<c_0(k)$ and one interval if $c > c_0(k)$. It remains to determine the position of these intervals with respect to the origin. The idea here is to look at the value of $\Delta$ at $z=0$:
\begin{equation*}
\Delta(0) =  k \cdot\varepsilon =( k- 1)^2  k (c k^2-1)^2 [c^2 k  (k-1)^2  -
   2 c (k^2+k-2)+k].
\end{equation*}
The degree 4 polynomial (in $c$) $\Delta(0)$ has the following roots: $c=1/k^2$ (double root, corresponding to the atom at $0$),
\begin{align*}
c_1 = \frac{(\sqrt{k+1}-1)^2}{k(k-1)}, \\
c_2 = \frac{(\sqrt{k+1}+1)^2}{k(k-1)} .
\end{align*}
We conclude that, for $c \in [c_1(k),c_2(k)]$, $\Delta(0) \leq 0$, and thus $0$ belongs to the support of $\mu_{k,c}$. Direct computation shows that $1/k^2 \leq c_1 < c_2$ for all $k$. Also, one can show that $c_0 > c_1$, again for all $k>1$. Finally, the curves of $c_0$ and $c_2$ intersect at $k_0 \approx 13.637$. The plots of the functions $c_0,c_1$ and $c_2$ are presented in Figure \ref{fig:c-k}.
\begin{figure}
\centering
\subfigure[]{\includegraphics[width=0.45\textwidth]{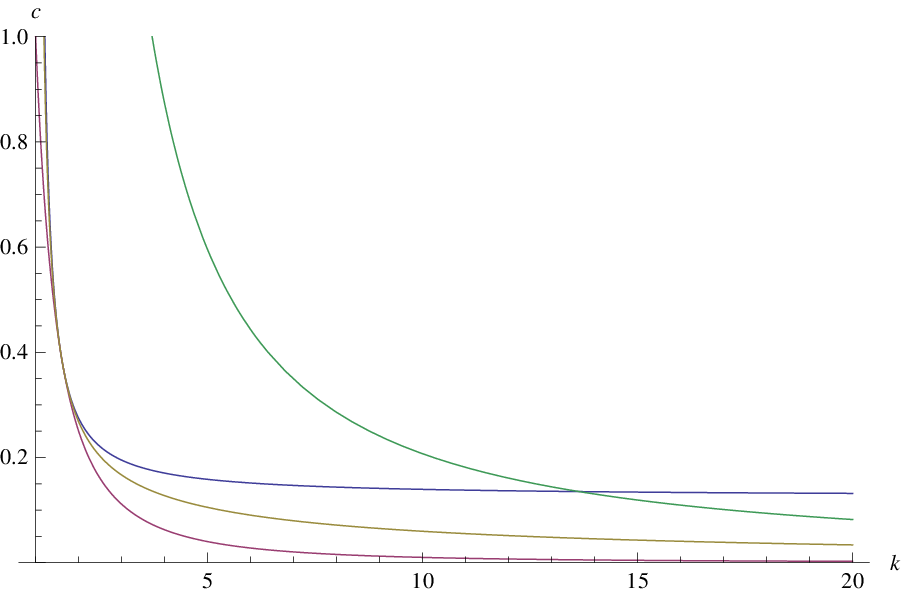}}\quad
\subfigure[]{\includegraphics[width=0.45\textwidth]{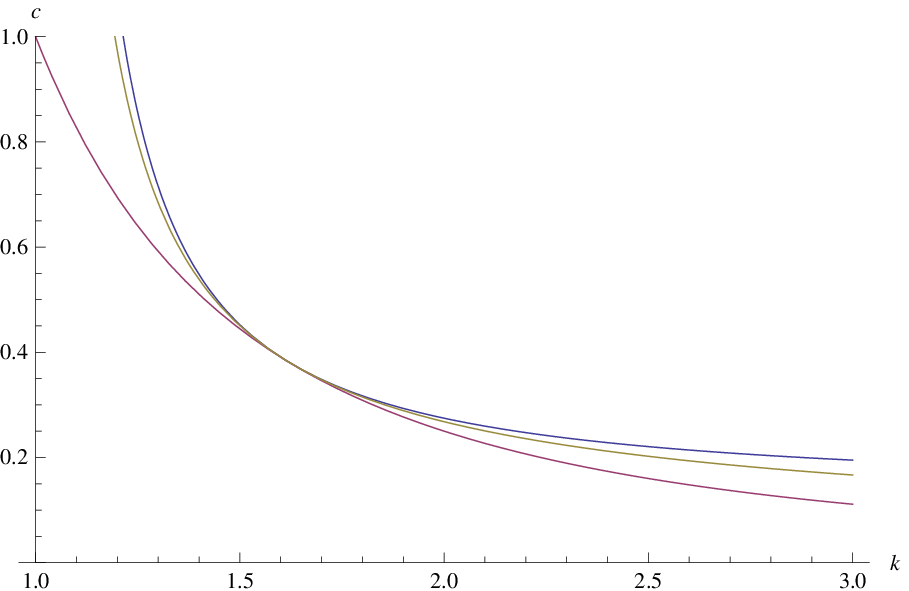}}
\caption{Graphs of the functions $1/k^2, c_0, c_1$ and $c_2$. On the left, from bottom to top at $k=20$, we have plotted $1/k^2$ (red), $c_1$ (yellow), $c_2$ (green) and $c_0$ (blue). On the right, a detail of the graph around the touching point $k=(\sqrt 5 +1)/2$.}
\label{fig:c-k}
\end{figure}

A schematic representation of the regions delimited by the four curves $c_{0,1,2}$ and $k \mapsto 1/k^2$ is presented in Figure \ref{fig:phase-diagram}.

\begin{figure}
\centering
\includegraphics{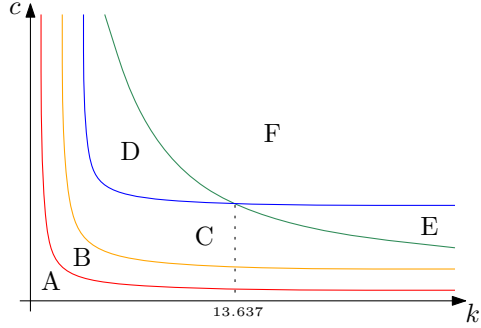}
\caption{Schematic phase diagram for the spectrum of $\mu_{k,c}$. From bottom to top at the right hand side of the graph, we have plotted $1/k^2$ (red), $c_1$ (yellow), $c_2$ (green) and $c_0$ (blue).}
\label{fig:phase-diagram}
\end{figure}

Note that we do not represent the touching point $k=(\sqrt 5 +1)/2$ (which is the golden ratio) appearing in Figure \ref{fig:c-k}. Six different regions are delimited by the curves and we summarize the properties of the support of the measure $\mu_{k,c}$ for $(k,c)$ lying in different regions in Table \ref{tbl:regions}. The sign of $\Delta_2$ is given by the position of the region with respect to the curve $c_0$, which translates in turn to the support having one or two intervals. The sign of $\Delta(0)$ is negative iff $0$ belongs to the support of $\mu_{k,c}$ and this happens for regions situated between $c_1$ and $c_2$. Finally, only measures having parameters situated below the graph of $k \mapsto 1/k^2$ have an atom at $0$.

From the above considerations, it follows that the regions $A$, $C$ and $D$ correspond to parameters $(k,c)$ for which the support of $\mu_{k,c}$ is not strictly positive. Moreover, since for region $F$, the support has a single interval which does not contained $0$, this interval must be situated on the positive half line (otherwise, the measure would be supported on the negative half line, which is impossible, since the average of $\mu_{k,c}$ is positive); thus, region $F$ corresponds to measures having positive support.

Measures with parameters in regions $B$ and $E$ have support made of two intervals, none of which contains $0$. By the same argument as before, these intervals cannot be both contained in the negative half line, so either they are both positive or $0$ separates them.
In the following we study the positivity of the roots of $\Delta=\Delta(z)$ in these regions. We note that its leading coefficient $\alpha$ is positive and since $c<c_1(k)$ in region $B$ and $c>c_2(k)$ in region $E$, we have that $\varepsilon>0$.
Moreover, in these regions $\Delta$ has four real roots ($c<c_0(k)$), all different from zero ($\varepsilon>0$). We
prove that  the roots are positive if and only if $\beta<0$, $\gamma>0$ and $\delta<0$.
Indeed, the necessity follows by using Vi\`{e}te's formulas.  Conversely,
we have that
$$\Delta(-z)=k( \alpha z^4 - \beta z^3 + \gamma z^2 - \delta z + \varepsilon)$$
has no sign differences between consecutive coefficients and thus, by Descartes' rule of signs, it has no negative root.  This implies that all the four real roots must be positive. Therefore, the study of the positivity of the roots of $\Delta$ is reduced to the study of the sign of the coefficients $\beta$, $\gamma$ and $\delta$.

A simple computation shows that $\beta<0$ if and only if $c>p_0(k)$, where
$$p_0(k)=\frac{k(k-2)}{2(k-1)^2(k+1)}$$ is the solution of the equation $\beta=0$.

The equation (in $c$) $\gamma=0$ is quadratic, with positive dominant term, and
$$\gamma(p_0) = -\frac{k \left(k^6+10 k^5-24 k^4-2 k^3+16 k^2-12 k+12\right)}{2 (k-1)^2 (k+1)^2}.$$
The degree $6$ polynomial in the numerator has $4$ real roots, all smaller than $2$ (as it can be checked by a computer \cite{mat}), hence $\gamma(p_0)<0$ for all $k \geq 2$. It follows that
\begin{equation*}\label{eq:gam}
p_1(k)<p_0(k)<p_2(k), \text{ for  all } k\geq 2,
\end{equation*}
where $p_{1,2}$ are the solutions of the equation $\gamma=0$:
\begin{equation*}
p_{1,2}(k)=\frac{3 k^3-k^2-3 k+6\pm\sqrt{3 k^6+30 k^5-45 k^4-6 k^3+45 k^2-36 k+36}}{2 k (3k^2-4)(k-1)}.
\end{equation*}
Finally, the discriminant of the cubic equation (in $c$) $\delta=0$ is
$$\Delta_{\delta}=16 (k-1)^6 k^2 (k^2-k-1)^2 P_{10}(k),$$
where $P_{10}$ is a polynomial of degree 10 with positive leading coefficient, which has six real roots, all smaller than 2 and thus $\Delta_{\delta}>0$ for $k\geq 2$. This fact implies that the equation $\delta=0$ has three real solutions, denoted by $f_1$, $f_2$ and $f_3$. We assume that $f_1<f_2<f_3$. Since $f_1+f_2+f_3>0$ and $f_1 f_2 f_3 \leq 0$ for all $k\geq 2$, it follows that $f_1\leq 0$ and $f_1,f_2>0$. We conclude that $\delta<0$ if and only if  $c\in(0,f_2(k))\cup(f_3(k),\infty)$.

In the following we prove that $f_2<p_2<f_3$. Since $p_2$ is positive, this is equivalent to $\delta(p_2)>0$. However, the expression for $\delta(p_2)$ is too complicated, so we shall lower and upper bound $p_2$ by simpler quantities
\begin{equation}\label{eq:dd0}
0<p_2^{s}(k)<p_2(k)<p_2^l(k), \text{ for all } k\geq 2,
\end{equation}
where
$$p_2^s(k):=\frac{3k^3-k^2-3k+6}{2k(3k^2-4)(k-1)} \text{ and } p_2^l:=\frac{6k^3-k^2-3k+6}{2k(3k^2-4)(k-1)}.$$
The values of $\delta$ at the points $p_2^s$ and $p_2^l$ are given by
$$\delta(p_2^m)=\frac{1}{2k(3k^2-4)^3}P_{10}^m(k),\;m=s,l,$$
where $P_{10}^s$ and $P_{10}^l$ are polynomials of degree 10 with positive leading coefficient, which have four real roots, all smaller than 2.
Therefore,
\begin{equation}\label{eq:dd2}
\delta(p_2^m)>0, \text{ for } k\geq 2, \,m=s,l.
\end{equation}
By \eqref{eq:dd0}, \eqref{eq:dd2} and taking into account that $\delta(c)\to -\infty$ as $c\to\infty$, for $k\geq 2$, we have that $p_0<f_2<p_2^s<p_2<p_2^l<f_3$. Now, we can conclude that $\Delta$ has four  positive roots if and only if $c>f_3$.

In a similar manner, we can prove that
$$c_1<\frac{k+1}{k(k-1)}<f_3<\frac{k+4}{k(k-1)}<c_2, \text{ for } k\geq 2.$$

It proves that the region $B$ contains parameters for which the support of $\mu_{k,c}$ consists of two intervals, one negative and the other one positive, and in the region $E$ the support of $\mu_{k,c}$ consists of two positive intervals.

\begin{table}[htdp]
\caption{Properties of the support of $\mu_{k,c}$ in terms of the parameters.}
\begin{center}
\begin{tabular}{|c|c|c|c|c|c|}
\hline
Region & $\Delta_2$ & \# of intervals & $0 \in \text{support}$ & Atom at $0$ & $\text{Support} \subset (0,\infty)$\\
\hline\hline
A & + & 2 & No  & Yes & No  \\  \hline
B & + & 2 & No  & No  & No  \\  \hline
C & + & 2 & Yes & No  & No  \\  \hline
D & - & 1 & Yes & No  & No  \\  \hline
E & + & 2 & No  & No  & Yes \\  \hline
F & - & 1 & No  & No  & Yes \\  \hline
\end{tabular}
\end{center}
\label{tbl:regions}
\end{table}

We conclude that the regions of parameters corresponding to measures $\mu_{k,c}$ supported on the open positive half line are $E$ and $F$, which correspond to values of $c$ satisfying
$$c>c_2=:c_{red}.$$

\end{proof}

\begin{remark}
The eventual atom at $0$ in the above theorem can be understood via Proposition \ref{prop:rank}: if $s = cnk$ and $c <1/k^2$ (which is equivalent to $s<n/k$), then the matrix $R$ will have at least $nk-k^2s$ zero eigenvalues, i.e. a fraction $1-ck^2$ of its total number of eigenvalues. Thus, its empirical eigenvalue distribution will have a Dirac mass at $0$, of mass at least $1-ck^2>0$.
\end{remark}

Going back to our original motivation, we state now, as a consequence of Theorem \ref{thm:positivity-mu-kc}, a result about a threshold for the reduction criterion, in the spirit of Section \ref{sec:Wishart}. Recall that a \emph{threshold} for the value $c$ of the parameter giving the scaling of the environment $s \sim cnk$ is the value at which a sharp phase transition occurs: for values of $c$ larger than the threshold, when $n \to \infty$, quantum states will ``satisfy'' the reduction criterion with probability close to one, whereas for values of $c$ smaller than the threshold, the probability of a quantum state satisfying the criterion will be close to zero. The case of the simultaneous reduction criterion follows along the lines of Proposition \ref{prop:simultaneous-balanced}.

\begin{proposition}\label{prop:red-threshold}
Consider random quantum states distributed along the probability measure of random induces states $\rho_n \sim \chi_{nk,cnk}$ for some constants $c$ and $k$ and some growing parameter $n \to \infty$. Then, the value
$$c_{red}(k)= \frac{(1+\sqrt{k+1})^2}{k(k-1)} $$
 is a \emph{threshold for the reduction criterion}, in the following sense (below, $\mathbb P_n$ denotes the probability distribution of $\rho_n$):
\begin{enumerate}
\item For all $k$ and $c<c_{red}(k)$,
$$\lim_{n \to \infty} \mathbb P_n\left[ \rho_n^{red} \geq 0 \right] =0.$$
\item For all $k$ and $c>c_{red}(k)$,
$$\lim_{n \to \infty} \mathbb P_n\left[ \rho_n^{red} \geq 0 \right] =1.$$
\end{enumerate}

In the case of the \emph{simultaneous} reduction criterion, let $m:=\min(n,k)$ be a fixed parameter and $\max(n,k) \to \infty$. Then, for the same linear scaling of the environment $s \sim cnk$, we have that the value
$$c_{red}(m) =  \frac{(1+\sqrt{m+1})^2}{m(m-1)} $$
is a threshold for the simultaneous reduction criterion, in the same sense as above.
\end{proposition}

\section{Comparing entanglement criteria via thresholds}
\label{sec:vs-other}

In this final section, we would like to compare the reduction criterion with other entanglement criteria, via thresholds. We gather in Table \ref{tbl:threshold-comp} results about different thresholds for entanglement and entanglement criteria, taken from the existent literature \cite{aub,ane,asy}. In the table, the scaling of the environment dimension is $s \sim cnk$, and the threshold are defined to be the sharp values of $c$ for which there is a phase transition between zero and unit asymptotic probabilities. The thresholds in the first line of the table have been obtained by Aubrun et al. \cite{asy}: the size $s_0$ of the environment for which the phase transition occurs has been bounded as follows:
$$C_1 nk \min(n,k) \leq s_0 \leq C_2 nk \log^2 (nk) \min(n,k).$$

The contribution of this work is the last line of the table, where thresholds have been obtained in the three regimes. Only in the second unbalanced regime, a non-trivial threshold has been obtained (this is to be contrasted with the case of the stronger PPT criterion).

The threshold for the partial transposition criterion, in the unbalanced cases, follows from the paper of Banica and Nechita \cite{bne12}:
\begin{equation*}
c_{PPT} = 2+2\sqrt{1 - \frac{1}{k^2}}.
\end{equation*}
Note that this value is always larger than $c_{red}$ (see \eqref{eq:threshold-red}), with equality if and only if $k=2$. This is in agreement with the fact that the reduction criterion is in general weaker that the PPT criterion (since the reduction map is co-completely positive) and the fact that for quantum states in $\mathbb{C}^n \otimes \mathbb{C}^2$  quantum systems, the two criteria are equivalent ($c_{red}(2) = c_{PPT}(2)$). The fact that, for $k \geq 3$, $c_{red} < c_{PPT
}$ implies the following striking phenomenon: for any intermediate value $c_{red} < c < c_{PPT}$, for fixed $k$ and large $n$, with large $\chi_{nk,cnk}$-probability, a random density matrix will satisfy the reduction criterion, while the PPT criterion will detect its entanglement. It is also worthwhile to remark how the thresholds for the two criteria scale for large values of $k$: $c_{red} \to 0$, while $c_{PPT} \to 4$ as $k \to \infty$. This hints to the situation in the balanced regime, where $c_{red} = 0$, while for the PPT criterion, it has been shown by Aubrun \cite{aub} that $c_{PPT}=4$, see Table \ref{tbl:threshold-comp}.

\begin{table}[htdp]
\begin{center}
\renewcommand{\arraystretch}{1.7}
\begin{tabular}{|c||c|c|c|}
\hline
Criterion \textbackslash \,  Regime & Balanced ($n = k \to \infty$) & First unbalanced ($k \to \infty$) &  Second unbalanced  ($n \to \infty$) \\ \hline \hline
Entanglement & $\infty$, $(\sim n\log^q n)$&  $\sim n k$ & $\sim n k$ \\ \hline
Partial transp. & $4$ & $2 + 2\sqrt{1-\frac{1}{n^2}}$ & $2 + 2\sqrt{1-\frac{1}{k^2}}$\\ \hline
Realignment & $(8/3\pi)^2 \approx 0.72$ & 0  & 0\\ \hline
Reduction & $0$ & $0$ & $\frac{(\sqrt{k+1}+1)^2}{k(k-1)}$ \\ \hline
\end{tabular}
\end{center}
\caption{Comparing thresholds for entanglement  and entanglement criteria in different asymptotical regimes.}
\label{tbl:threshold-comp}
\end{table}

We present in Figure \ref{fig:simulation} some numerical simulations results for the second unbalanced regime. We consider two values for the parameter $c$, one below and the other above the threshold \eqref{eq:threshold-red}. The plots are eigenvalue histograms for $R=W^{red}$ for one realization of a random matrix $W$.

\begin{figure}
\centering
\subfigure[]{\includegraphics[width=0.45\textwidth]{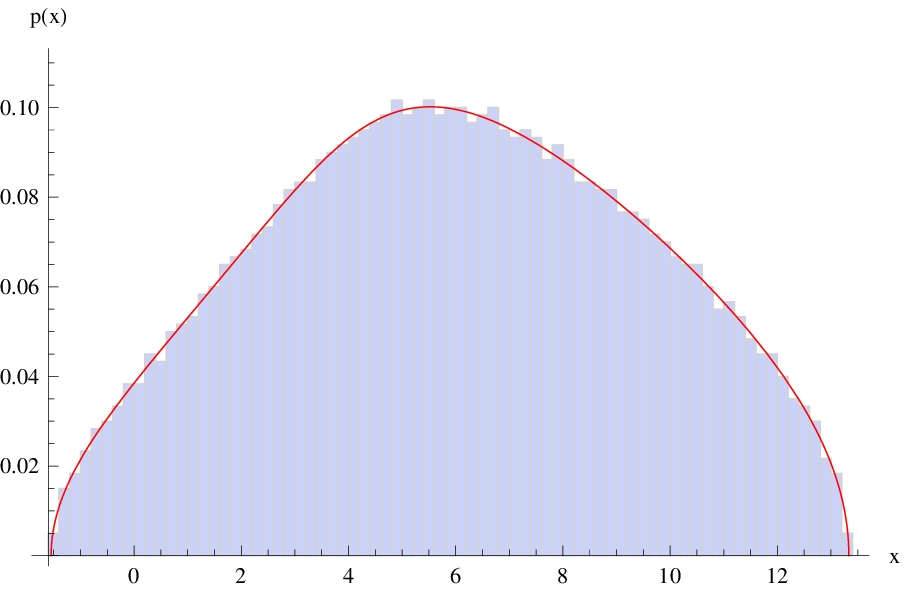}}\qquad
\subfigure[]{\includegraphics[width=0.45\textwidth]{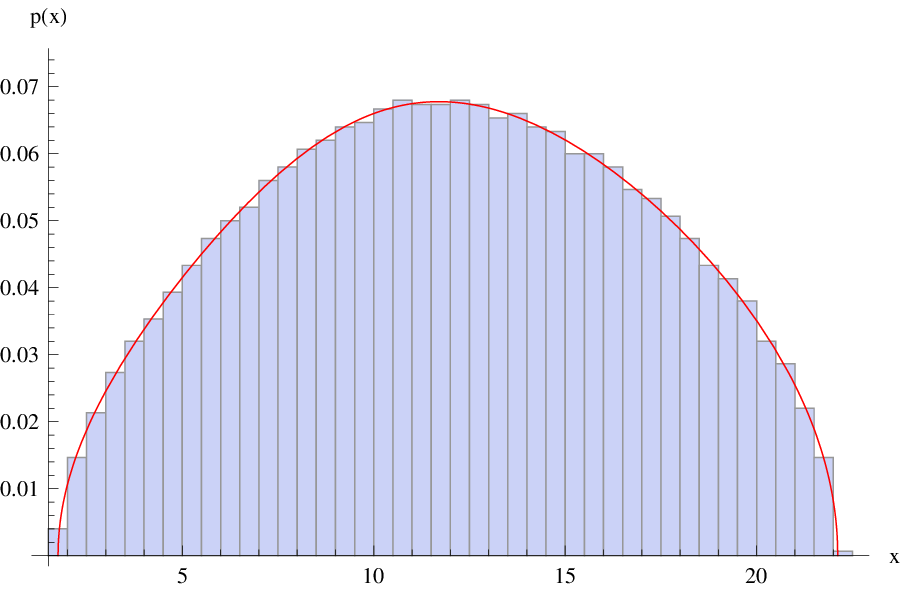}}
\caption{Numerical simulations for the (normalized) eigenvalues of a reduced Wishart matrix. We consider only the second unbalanced regime, with $n=1000$ and $k=3$, and different environment scaling. On the left, $c=1 < c_{red}=3/2 $; on the right, $c=2>c_{red}$. The numerics (blue histogram) are compared with the theoretical curve for the density (red curve) obtained by Stieltjes inversion from the corresponding Cauchy transform.}
\label{fig:simulation}
\end{figure}

\appendix
\section{A combinatorial lemma}
\label{sec:lemma-combinatorial}

In this appendix, we prove the following combinatorial lemma, used in the proof of Theorem \ref{thm:unbalanced-B}.

\begin{lemma}\label{lem:Pfm-alpha}
For any geodesic permutation $\alpha\in \mathcal S_{NC(\gamma)}$ and any
choice function $f \in \mathcal F_p$, we have
\begin{equation*}
\#(P_f^{-1}\alpha) = p-|f^{-1}(1)| + 2\sum_{b \in \alpha}
\mathbf{1}_{f_b  \equiv 1} + 1 -\# \alpha -\mathbf 1_{f \equiv 1},
\end{equation*}
where we denote by $f_b$ the restriction of  $f$ to a cycle $b$ of
$\alpha$.
\end{lemma}
\begin{proof}
We use a recurrence over the number of cycles of $\alpha$. If
$\alpha$ has just one cycle (i.e. $\alpha=\gamma$), then  we have
that
$$\sum_{b \in \alpha}\mathbf{1}_{f_b  \equiv
1}=\mathbf{1}_{f\equiv 1}$$ and by Lemma \ref{lem.Pf} (iii) we
obtain the conclusion.

For partitions $\alpha$ with more than one cycle it is convenient
to identify $\alpha$ with a non-crossing partition and we can
assume without loss of generality that $\alpha = \hat 1_r \oplus
\beta$, where $\hat 1_r$ is the contiguous block of size $r$,
$r\in\{1,\ldots,p-1\}$.

Let us first introduce some notations, which we used in what
follows. For a function $f\in\mathcal F_p$, we denote by $g$ the
restriction of $f$ to the set $\{1,\ldots,r\}$ and by $h$ the
restriction of $f$ to  $\{r+1,\ldots,p\}$. We  simply write
$f=g\oplus h$. By Definition \ref{def:Pf} we  have that $P_g\in
\mathcal S_r$ and $P_h\in \mathcal S(\{r+1,\ldots,p\})$. Now, we can define a
permutation $P_{g}\oplus P_{h}\in \mathcal S _p$  as follows
\begin{equation*}
\left(P_{g}\oplus P_{h} \right)(i)=\begin{cases}
P_{g}(i), &\text{for } i\in\{1,\ldots,r\},\\
P_{h}(i), &\text{for } i\in\{r+1,\ldots,p\}.
\end{cases}
\end{equation*}

In the following we want to obtain a formula for
$\#\left(P_{g\oplus h}^{-1}\alpha\right)$ in terms of
$\#\left(P_{g}^{-1}\gamma_{r}\right)$ and
$\#\left(P_{h}^{-1}\beta\right)$, where $\gamma_r$ is the
restriction of the full cycle $\gamma$ to a block of $\alpha$
containing  $r$ elements.
We distinguish the following three cases: \\
1)  If $g(i)= 1$, for all $i\in\{1,\ldots,r\}$,  then by
Definition \ref{def:Pf}, we have that
\begin{align} \label{eq:c1}
 \#\left(P_{f}^{-1}\alpha\right)&=\#\left(\left(\mathrm{id}\oplus P_{h}^{-1}\right)\left(\gamma_r\oplus
 \beta\right)\right) \notag\\
 &=\#\left(\gamma_r\oplus P_{h}^{-1}\beta\right)\\
 &= 1+\#\left(P_{h}^{-1}\beta\right).\notag
\end{align}
2) If $h(i)= 1$, for all $i\in\{r+1,\ldots,p\}$,  then we get
\begin{align} \label{eq:c2}
 \#\left(P_{f}^{-1}\alpha\right)&=\#\left(\left(P_{g}^{-1}\oplus \mathrm{id}\right)\left(\gamma_r\oplus
 \beta\right)\right) \notag\\
 &=\#\left(P_{g}^{-1}\gamma_{r}\right)+\#\beta.
\end{align}
3) If there exist $i_{0}\in\{1,\ldots,r\}$ and
$j_{0}\in\{r+1,\ldots,p\}$ such that $g(i_0)\neq 1$ and
$h(j_0)\neq 1$, then we prove that
\begin{equation}\label{eq:c3}
 \#\left(P_{f}^{-1}\alpha\right)=\#\left(P_{g}^{-1}\gamma_{\,r}\right)+\#\left(P_{h}^{-1}\beta\right)-1.
\end{equation}
Indeed, since $g^{-1}(2)\neq\emptyset$ and
$h^{-1}(2)\neq\emptyset$, the constants
\begin{align*}
&a=\min g^{-1}(2),\;b=\max g^{-1}(2)\\
&c=\min h^{-1}(2),\;d=\max h^{-1}(2)
\end{align*}
are well defined and we have that $a\leq b<c\leq d$. Moreover, a
simple computation shows that
$P_{g}=(a\, b \ldots)$ and $P_{h}=(c\, d \ldots)$ (if $a=b$ or/and $c=d$, then $P_g=\mathrm{id}$ or/and $P_h=\mathrm{id}$).
Since  $P_{g\oplus h}=(a\, d \ldots c\, b \ldots)$, it follows
that
\begin{equation*}\label{eq:ac}
P_{g\oplus h}=\left(P_g\oplus P_h\right) (a\, c).
\end{equation*}
Using the relation above and Eq. \eqref{eq:tr}, we
have
\begin{align*}
 \#\left(P_{g\oplus h}^{-1}\alpha\right)&=\#\left((a\,c) \left(P_{g}^{-1}\oplus P_{h}^{-1}\right)\left(\gamma_{\, r}\oplus \beta\right)
 \right)\\
 &=\#\left(\left(P_{g}^{-1}\oplus P_{h}^{-1}\right)\left(\gamma_{\, r}\oplus
 \beta\right)\right)\pm 1.
\end{align*}

Since $a$ and $c$  belong to different cycles of
$\left(P_{g}^{-1}\oplus P_{h}^{-1}\right)\left(\gamma_{\, r}\oplus
\beta\right)$, it follows that
\begin{align*}
  \#\left(P_{g\oplus h}^{-1}\alpha\right)&=\#\left(P_{g}^{-1}\gamma_{\, r}\oplus
 P_{h}^{-1}\beta\right)-1\\
 &=\#\left(P_{g}^{-1}\gamma_{\, r}\right)+\#\left(
 P_{h}^{-1}\beta\right)-1,
\end{align*}
and thus \eqref{eq:c3} hods.

Combining relations \eqref{eq:c1}, \eqref{eq:c2} and
\eqref{eq:c3}, we observe that
\begin{equation*}\label{eq.f1}
\#\left(P_{g\oplus h}^{-1}\alpha\right)=
\#\left(P_{g}^{-1}\gamma_{r}\right)+\#\left(P_{h}^{-1}\beta\right)-
\left(1-\textbf{1}_{g\equiv 1}\right)\left(1- \textbf{1}_{h\equiv
1}\right),
\end{equation*}
which is equivalent to
\begin{equation}\label{eq.f2}
\#\left(P_{g\oplus h}^{-1}\alpha\right)=
\#\left(P_{g}^{-1}\gamma_{r}\right)+\#\left(P_{h}^{-1}\beta\right)-
\left(1-\textbf{1}_{g\equiv 1}- \textbf{1}_{h\equiv
1}+\textbf{1}_{g\oplus h\equiv 1} \right).
\end{equation}

Successively applying Eq. \eqref{eq.f2}, we get
\begin{align*}
\#\left(P_{f}^{-1}\alpha\right)&=
\#\left(P_{g}^{-1}\gamma_{r}\right)+\#\left(P_{h}^{-1}\beta\right)-
\left(1-\textbf{1}_{g\equiv 1}- \textbf{1}_{h\equiv
1}+\textbf{1}_{f\equiv 1} \right)\\
&=\#\left(P_{g}^{-1}\gamma_{r}\right)+\#\left(P_{g'\oplus
h'}^{-1}\left( \hat{1}_{r'} \oplus \beta'\right)\right)-
\left(1-\textbf{1}_{g\equiv 1}-
\textbf{1}_{h\equiv 1}+\textbf{1}_{f\equiv 1} \right)\\
&=\#\left(P_{g}^{-1}\gamma_{r}\right)+\#\left(P_{g'}^{-1}\gamma_{r'}\right)+\#\left(P_{h'}^{-1}\beta'\right)-
\left(2-\textbf{1}_{g\equiv 1}-\textbf{1}_{g'\equiv 1}-
\textbf{1}_{h'\equiv 1}+\textbf{1}_{f\equiv 1} \right)\\
&\;\;\vdots\\
&=\sum\limits_{b\in\alpha}\#\left(P_{f_b}^{-1}\gamma_{|b|}\right)-\left(\#\alpha
-1 -\sum\limits_{b\in\alpha}\textbf{1}_{f_b\equiv
1}+\textbf{1}_{f\equiv 1}\right).
\end{align*}
Therefore, we obtain
\begin{equation}\label{eq.f3}
\#\left(P_{f}^{-1}\alpha\right)=\sum\limits_{b\in\alpha}\#\left(P_{f_b}^{-1}\gamma_{|b|}\right)-\left(\#\alpha
-1 -\sum\limits_{b\in\alpha}\textbf{1}_{f_b\equiv
1}+\textbf{1}_{f\equiv 1}\right).
\end{equation}
Using Lemma \ref{lem.Pf} (iii) in relation \eqref{eq.f3} and taking
into account the identities
$$ \sum\limits_{b\in\alpha}|b|=p \text{ and }
\sum\limits_{b\in\alpha}|f_b^{-1}(1)|=|f^{-1}(1)|,$$
the conclusion follows.
\end{proof}

\begin{remark}
The statement of the theorem above is not valid if $\alpha$ is not
a geodesic permutation.  Indeed, setting $\alpha=(1\,3)(2\, 4)\in
\mathcal S_4$ which is not a geodesic permutation and considering the function
$f\in \mathcal{F}_{4}$ given by $f(1)=f(3)=1$ and $f(2)=f(4)=2$,
it follows that $P_f=P_f^{-1}=(2 \, 4)$ and thus
$\#\left(P_f^{-1}\alpha\right)=\# (1\, 3)=1$. On the other hand,
we have
$$p-|f^{-1}(1)| + 2\sum_{b \in\alpha}
\mathbf{1}_{f_b  \equiv 1} + 1 -\# \alpha -\mathbf 1_{f \equiv
1}=4-2+2(1+0)+1-2-0=3.$$
\end{remark}

\section{Some convergence results}
\label{A-convergence}

We start with the classical Bai-Yin result \cite{byi}:
\begin{proposition}\label{prop:Bai-Yin}
Let $W_n$ be a  Wishart matrix of  parameters $n$ and $s_n$, such that $s_n / n \to c$ as $n\to\infty$, for some positive constant $c$.  Then, almost surely,
$$\lim_{n \to \infty} \left\|\frac{W_n}{s_n}\right\| = (\sqrt c +1)^2.$$
\end{proposition}

We recall now a result of Collins et al.  \cite{cny} and we prove a useful bound:

\begin{proposition}\label{prop:Z-d}
Let $W_n$ be a   Wishart matrix of  parameters $n$ and $s_n$, such that $s_n / n \to \infty$ as $n\to\infty$. Put
$$Z_n = \sqrt{ns_n}\left(\frac{W_n}{ns_n} - \frac{\mathrm{I}_n}{n} \right),\,n=1,2,\ldots$$
Then, the sequence $(Z_n)$ converges in moments to the standard semicircular distribution
$$d\sigma = \frac{1}{2\pi}\sqrt{4-x^2} \, \mathbf{1}_{[-2,2]}(x) dx,$$
where  $\mathbf{1}_{[-2,2]}$ is the indicator function of the interval $[-2,2]$.
Moreover, for any $\varepsilon > 0$, there exist some constants $C,a>0$, such that
$$\forall n, \qquad \mathbb P(\|Z_n\|>2+\varepsilon) \leq C\exp(-an^{1/3}).$$
In particular, almost surely,
$$\lim_{n \to \infty} \|Z_n\| = 2.$$
\end{proposition}
\begin{proof}
We only prove the upper bound, the other results being taken from the paper of Collins et al. \cite{cny} It was shown 
that, for all $\varepsilon>0$, there exists a constant $C>0$ such that, for all $n$ large enough for $\sqrt{n/s} < \varepsilon/2$ to hold, and for all $p < \sqrt n$,
$$\mathbb E \operatorname{Tr}(Z_n^p) \leq C (2+\varepsilon)^p.$$
Fix such a positive $\varepsilon$, and write
\begin{align*}
\mathbb P\left(\|Z_n\|> 2+2\varepsilon\right) &\leq \mathbb P\left(\operatorname{Tr}(Z_n^{2p}) > (2+2\varepsilon)^{2p}\right) \\
&\leq \frac{\mathbb E \operatorname{Tr}(Z_n^{2p})}{(2+2\varepsilon)^{2p}} \leq \frac{C(2+\varepsilon)^{2p}}{(2+2\varepsilon)^{2p}},
\end{align*}
where we used the bound from the paper of Collins et al. \cite{cny} for $\varepsilon/2$. The conclusion follows easily, by letting $p=\lceil n^{1/3} \rceil$ and adjusting the constant $C$ to accommodate for small values of $n$.
\end{proof}

\end{document}